\documentclass{article}

\usepackage{arxiv}

\usepackage[utf8]{inputenc} 
\usepackage{authblk}
\usepackage{graphicx}
\usepackage{amssymb}
\usepackage{amsmath}
\usepackage{amsthm}
\usepackage{algorithm}
\usepackage[noend]{algpseudocode}
\usepackage{titletoc}
\usepackage{titlesec}
\usepackage[utf8]{inputenc}
\usepackage{enumerate}

\makeatletter
\renewcommand{\Function}[2]{%
	\csname ALG@cmd@\ALG@L @Function\endcsname{#1}{#2}%
	\def\jayden@currentfunction{#1}%
}
\newcommand{\funclabel}[1]{%
	\@bsphack
	\protected@write\@auxout{}{%
		\string\newlabel{#1}{{\jayden@currentfunction}{\thepage}}%
	}%
	\@esphack
}
\makeatother

\newtheorem{thm}{Theorem}[section]
\newtheorem{lem}[thm]{Lemma}
\newtheorem{prop}[thm]{Proposition}

\newtheorem{defn}{Definition}[section]

\usepackage[english]{babel}

\newcommand{\dfnn}[2]{$\quad$ \textbf{\emph{(#1)}} {#2}}

\renewcommand{\vec}[1]{\mathbf{#1}}
\newcommand{\eself}{\hookrightarrow^{s}}
\newcommand{\eref}{\hookrightarrow^{r}}
\newcommand{\eancestor}{\hookrightarrow^{a}} 
\newcommand{\eselfancestor}{\hookrightarrow^{sa}} 
\newcommand{\erefz}{\hookrightarrow}
\newcommand{\efork}{\pitchfork}

\newcommand{\dom}{\gg}
\newcommand{\sdom}{\gg^{s}}
\newcommand{\domf}{\gg^{f}}

\newcommand{\hibefore}{\mapsto}
\newcommand{\hbefore}{\rightarrow}
\newcommand{\concur}{\parallel}

\makeatletter
\def\BState{\State\hskip-\ALG@thistlm}
\makeatother

\title{Fantom: A scalable framework for asynchronous distributed systems}

\author{\large Sang-Min Choi}
\author{Jiho Park}
\author{Quan Nguyen}
\author{Andre Cronje}

\affil{FANTOM Lab\\ FANTOM Foundation}

\begin{document}
\maketitle

\begin{abstract}
	We describe \emph{Fantom}, a framework for asynchronous distributed systems. \emph{Fantom} is based on the Lachesis Protocol~\cite{lachesis01}, which uses asynchronous event transmission for practical Byzantine fault tolerance (pBFT) to create a leaderless, scalable, asynchronous Directed Acyclic Graph (DAG).
	
	We further optimize the \emph{Lachesis Protocol} by introducing a permission-less network for dynamic participation. Root selection cost is further optimized by the introduction of an n-row flag table, as well as optimizing path selection by introducing domination relationships.
	
	We propose an alternative framework for distributed ledgers, based on asynchronous partially ordered sets with logical time ordering instead of blockchains.
	
	This paper builds upon the original proposed family of \emph{Lachesis-class} consensus protocols. We formalize our proofs into a model that can be applied to abstract asynchronous distributed system.
\end{abstract}


\newpage
\pagenumbering{arabic} 
\tableofcontents 
\newpage
\section{Introduction}\label{ch:intro}

Blockchain has emerged as a technology for secure decentralized  transaction ledgers with broad applications in financial systems, supply chains and health care.
\emph{Byzantine} fault tolerance~\cite{Lamport82} is addressed in distributed database systems, in which up to one-third of the participant nodes may be compromised. Consensus algorithms~\cite{bcbook15} ensures the integrity of transactions between participants over a distributed network~\cite{Lamport82} and is equivalent to the proof of \emph{Byzantine} fault tolerance in distributed database systems~\cite{randomized03, paxos01}. 

A large number of consensus algorithms have been proposed. For example; 
the original Nakamoto consensus protocol in Bitcoin uses Proof of Work (PoW)~\cite{bitcoin08}. Proof Of Stake (PoS)~\cite{ppcoin12,dpos14} uses participants' stakes to generate the blocks respectively. Our previous paper gives a survey of previous DAG-based approaches~\cite{lachesis01}.

In the previous paper~\cite{lachesis01},  we introduced a new consensus protocol, called $L_0$. The protocol $L_0$ is a DAG-based asynchronous non-deterministic protocol that guarantees pBFT.
$L_0$ generates each block asynchronously and uses the OPERA chain (DAG) for faster consensus by confirming how many nodes share the blocks.

The Lachesis protocol as previously proposed is a set of protocols that create a directed acyclic graph for distributed systems. Each node can receive transactions and batch them into an event block. An event block is then shared with it's peers. When peers communicate they share this information again and thus spread this information through the network. In BFT systems we would use a broadcast voting approach and ask each node to vote on the validity of each block. This event is synchronous in nature. Instead we proposed an asynchronous system where we leverage the concepts of distributed common knowledge, dominator relations in graph theory and broadcast based gossip to achieve a local view with high probability of being a global view. It accomplishes this asynchronously, meaning that we can increase throughput near linearly as nodes enter the network.

In this work, we propose a further enhancement on these concepts and we formalize them so that they can be applied to any asynchronous distributed system.

\subsection{Contributions}

In summary, this paper makes the following contributions:
\begin{itemize}
\item We introduce the n-row flag table for faster root selection of the Lachesis Protocol.
\item We define continuous consistent cuts of a local view to achieve consensus.
\item We present proof of how domination relationships can be used for share information.
\item We formalize our proofs that can be applied to any generic asynchronous DAG solution.
\end{itemize}

\subsection{Paper structure}

The rest of this paper is organised as follows.
Section~\ref{se:prelim} describes our Fantom framework.
Section~\ref{se:lca} presents the protocol implementation.
Section~\ref{se:con} concludes.
Proof of Byzantine fault tolerance is described in Section~\ref{se:proof}.

\section{Preliminaries}\label{se:prelim}

The protocol is run via nodes representing users' machines which together create a network. The basic units of the protocol are called event blocks - a data structure created by a single node to share transaction and user information with the rest of the network. These event blocks reference previous event blocks that are known to the node. This flow or stream of information creates a sequence of history.

The history of the protocol can be represented by a directed acyclic graph $G=(V, E)$, where $V$ is a set of vertices and $E$ is a set of edges. Each vertex in a row (node) represents an event. Time flows left-to-right of the graph, so left vertices represent earlier events in history. 

For a graph $G$, a path $p$ in $G$ is a sequence  of vertices ($v_1$, $v_2$, $\dots$, $v_k$) by following the edges in $E$.
Let $v_c$ be a vertex in $G$.
A vertex $v_p$ is the \emph{parent} of $v_c$ if there is an edge from $v_p$ to $v_c$.
A vertex $v_a$ is an \emph{ancestor} of $v_c$ if there is a path from $v_a$ to $v_c$.

\begin{figure}[H] \centering
\includegraphics[height=5cm, width=1.0\columnwidth]{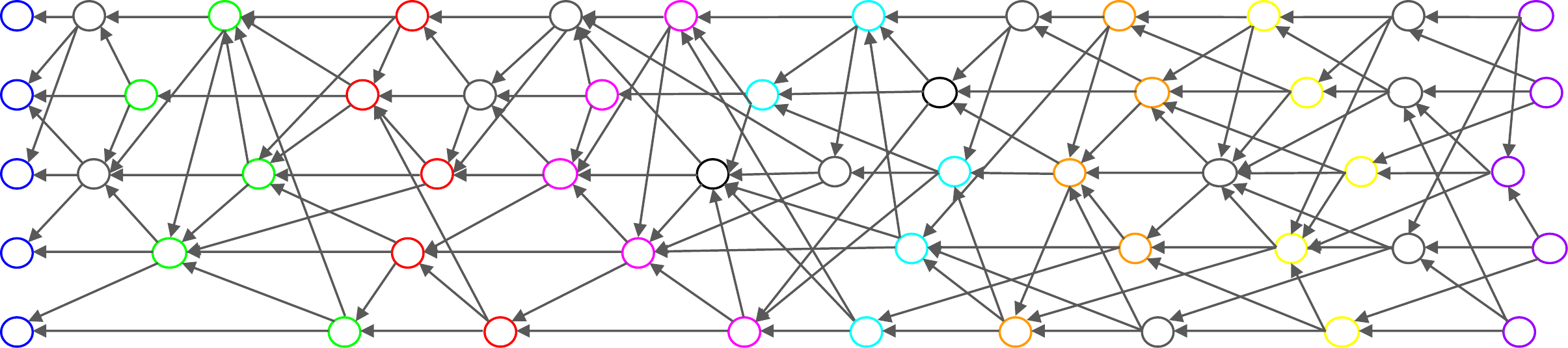}
\caption{An Example of OPERA Chain}
\label{fig:operachain}
\end{figure}

Figure~\ref{fig:operachain} shows an example of an OPERA chain (DAG) constructed through the Lachesis protocol. Event blocks are representated by circles. Blocks of the same frame have the same color.

\subsection{Basic Definitions}

\dfnn{Lachesis}{The set of \emph{protocols}}

\dfnn{Node}{Each machine that participates in the Lachesis protocol is called a \emph{node}. Let $n$ denote the total number of nodes.}

\dfnn{$k$}{A \emph{constant} defined in the system.}

\dfnn{Peer node}{A node $n_i$ has $k$ \emph{peer nodes}.}

\dfnn{Process}{A process $p_i$ represents a machine or a \emph{node}. The process identifier of $p_i$ is $i$. A set $P$ = \{1,...,$n$\} denotes the set of process identifiers.}

\dfnn{Channel}{A process $i$ can send messages to process $j$ if there is a channel ($i$,$j$). Let $C$ $\subseteq$ \{($i$,$j$) s.t. $i,j \in P$\} denote the set of channels.}

\dfnn{Event block}{Each node can create event blocks, send (receive) messages to (from) other nodes.The structure of an event block includes the signature, generation time, transaction history, and hash information to references.}

All nodes can create event blocks. The information of the referenced event blocks can be copied by each node. The first event block of each node is called a \emph{leaf event}.
 
Suppose a node $n_i$ creates an event $v_c$ after an event $v_s$ in $n_i$.  Each event block has exactly $k$ references. One of the references is self-reference, and the other $k$-1 references point to the top events of $n_i$'s $k$-1 peer nodes.

\dfnn{Top event}{An event $v$ is a top event of a node $n_i$ if there is no other event in $n_i$ referencing $v$.}

\dfnn{Height Vector}{The height vector is the number of event blocks \emph{created} by the $i$-th node.}

\dfnn{In-degree Vector}{The in-degree vector refers to the number of \emph{edges} from other event blocks created by other nodes to the top event block of this node. The top event block indicates the most recently created event block by this node.}

\dfnn{Ref}{An event $v_r$ is called ``ref" of event $v_c$ if the reference hash of $v_c$ points to the event $v_r$. Denoted by $v_c \eref v_r$. For simplicity, we can use $\erefz$ to denote a reference relationship (either $\eref$ or $\eself$).}

\dfnn{Self-ref}{An event $v_s$ is called ``self-ref" of event $v_c$, if the self-ref hash of $v_c$ points to the event $v_s$. Denoted by $v_c \eself v_s$.}

\dfnn{$k$ references}{Each event block has at least $k$ references. One of the references is self-reference, and the other $k$-1 references point to the top events of $n_i$'s $k$-1 peer nodes.}

\dfnn{Self-ancestor}{An event block $v_a$ is self-ancestor of an event block $v_c$ if there is a sequence of events such that $v_c \eself v_1 \eself \dots \eself v_m \eself v_a $. Denoted by $v_c \eselfancestor v_a$.}

\dfnn{Ancestor}{An event block $v_a$ is an ancestor of an event block $v_c$ if there is a sequence of events such that $v_c \erefz v_1 \erefz \dots \erefz v_m \erefz v_a $. Denoted by $v_c \eancestor v_a$.}

For simplicity, we simply use $v_c \eancestor v_s$ to refer both ancestor and self-ancestor relationship, unless we need to distinguish the two cases.

\dfnn{Flag Table}{The flag table is a $n \times k$ matrix, where $n$ is the number of nodes and k is the number of roots that an event block can reach. If an event block $e$ created by $i$-th node can reach $j$-th root, then the flag table stores the hash value of the $j$-th root.}

\subsection{Lamport timestamps}

Our Lachesis protocols relies on Lamport timestamps to define a topological ordering of event blocks in OPERA chain.  
By using Lamport timestamps, we do not rely on physical clocks to determine a partial ordering of events.

The ``happened before" relation, denoted by $\rightarrow$, gives a partial ordering of events from a distributed system of nodes. 
Each node $n_i$ (also called a process) is identified by its process identifier $i$. 
For a pair of event blocks $v$ and $v'$, the relation "$\rightarrow$" satisfies: (1) If $v$ and $v'$ are events of process $P_i$, and $v$ comes before $v'$, then $b \rightarrow v'$. (2) If $v$ is the send($m$) by one process and $v'$ is the receive($m$) by another process, then $v \rightarrow v'$. (3) If $v \rightarrow v'$ and $v' \rightarrow v''$ then $v \rightarrow v''$. 
Two distinct events $v$ and $v'$ are said to be concurrent if $v \nrightarrow v'$ and $v' \nrightarrow v$.

For an arbitrary total ordering $\prec$ of the processes, a relation $\Rightarrow$ is defined as follows: if $v$ is an event in process $P_i$ and $v'$ is an event in process $P_j$, then $v \Rightarrow v'$ if and only if either (i) $C_i(v) < C_j(v')$ or (ii) $C(v)= C_j(v')$ and $P_i \prec P_j$. This defines a total ordering, and that the Clock Condition implies that if $v \rightarrow v'$ then $v \Rightarrow v'$. 

We use this total ordering in our Lachesis protocol. This ordering is used to determine consensus time, as described in Section~\ref{se:lca}.

\dfnn{Happened-Immediate-Before}{An event block $v_x$ is said Happened-Immediate-Before an event block $v_y$ if $v_x$ is a (self-) ref of $v_y$. Denoted by $v_x \hibefore v_y$.}

\dfnn{Happened-before}{An event block $v_x$ is said Happened-Before an event block $v_y$ if $v_x$ is a (self-) ancestor of $v_y$. Denoted by $v_x \hbefore v_y$.}

Happened-before is the relationship between nodes which have event blocks. If there is a path from an event block $v_x$ to $v_y$, then $v_x$ Happened-before $v_y$. ``$v_x$ Happened-before $v_y$" means that the node creating $v_y$ knows event block $v_x$. This relation is the transitive closure of happens-immediately-before. Thus, an event $v_x$ happened before an event $v_y$ if one of the followings happens: (a) $v_y \eself v_x$, (b) $v_y \eref v_x$,  or (c) $v_y \eancestor v_x$. The happened-before relation of events form an acyclic directed graph $G' = (V, E')$ such that an edge $(v_i,v_j) \in E'$ has a reverse direction of the same edge in $E$.

\dfnn{Concurrent}{Two event blocks $v_x$ and $v_y$ are said concurrent if neither of them  happened before the other. Denoted by $v_x \concur v_y$.}

Given two vertices $v_x$ and $v_y$ both contained in two OPERA chains (DAGs) $G_1$ and $G_2$ on two nodes. We have the following:
\begin{enumerate}[(1)]
\item $v_x \hbefore v_y$ in $G_1$ if $v_x \hbefore v_y$ in $G_2$.
\item $v_x \concur v_y$ in $G_1$ if $v_x \concur v_y$ in $G_2$.
\end{enumerate}

\subsection{State Definitions}

Each node has a local state, a collection of histories, messages, event blocks, and peer information, we describe the components of each.

\dfnn{State}{A state of a process $i$ is denoted by $s_j^i$.}

\dfnn{Local State}{A local state consists of a sequence of event blocks $s_j^i = v_0^i, v_1^i, \dots, v_j^i$.}

In a DAG-based protocol, each $v_j^i$ event block is valid only if the reference blocks exist before it. From a local state $s_j^i$, one can reconstruct a unique DAG. That is, the mapping from a local state  $s_j^i$ into a DAG is \emph{injective} or one-to-one. Thus, for Fantom, we can simply denote the $j$-th local state of a process $i$ by the DAG $g_j^i$ (often we simply use $G_i$ to denote the current local state of a process $i$).

\dfnn{Action}{An action is a function from one local state to another local state.}

Generally speaking, an action can be one of: a $send(m)$ action where $m$ is a message, a $receive(m)$ action, and an internal action. A message $m$ is a triple $\langle i,j,B \rangle$ where $i \in P$ is the sender of the message, $j \in P$ is the message recipient, and $B$ is the body of the message. Let $M$ denote the set of messages. In the Lachesis protocol, $B$ consists of the content of an event block $v$.\\

Semantics-wise, in Lachesis, there are  two actions that can change a process's local state: creating a new event and receiving an event from another process.\\

\dfnn{Event}{An event is a tuple $\langle  s,\alpha,s' \rangle$ consisting of a state, an action, and a state. Sometimes, the event can be represented by the end state $s'$.}

The $j$-th event in history $h_i$ of process $i$ is $\langle  s_{j-1}^i,\alpha,s_j^i \rangle$, denoted by $v_j^i$.

\dfnn{Local history}{A local history $h_i$ of process $i$ is a (possibly infinite) sequence of alternating local states  --- beginning with a distinguished initial state. A set $H_i$ of possible local histories for each process $i$ in $P$.}

The state of a process can be obtained from its initial state and the sequence of actions or events that have occurred up to the current state. In the Lachesis protocol, we use append-only semantics. The local history may be equivalently described as either of the following:
$$h_i = s_0^i,\alpha_1^i,\alpha_2^i, \alpha_3^i \dots $$
$$h_i = s_0^i, v_1^i,v_2^i, v_3^i \dots $$
$$h_i = s_0^i, s_1^i, s_2^i, s_3^i, \dots$$

In Lachesis, a local history is equivalently expressed as:
$$h_i = g_0^i, g_1^i, g_2^i, g_3^i, \dots$$
where $g_j^i$ is the $j$-th local DAG (local state) of the process $i$.

\dfnn{Run}{Each asynchronous run is a vector of local histories. Denoted by
	$\sigma = \langle h_1,h_2,h_3,...h_N \rangle$.}

Let $\Sigma$ denote the set of asynchronous runs. We can now use Lamport’s theory to talk about global states of an asynchronous system. A global state of run $\sigma$ is an $n$-vector of prefixes of local histories of $\sigma$, one prefix per process. The happens-before relation can be used to define a consistent global state, often termed a consistent cut, as follows.

\subsection{Consistent Cut} 

Consistent cuts represent the concept of scalar time in distributed computation, it is possible to distinguish between a ``before'' and an ``after''.

In the Lachesis protocol, an OPERA chain $G=(V,E)$ is a directed acyclic graph (DAG). $V$ is a set of vertices and $E$ is a set of edges. DAG is a directed graph with no cycle. There is no path that has source and destination at the same vertex.
A path is a sequence of vertices ($v_1$, $v_2$, ..., $v_{k-1}$, $v_k$) that uses no edge more than once. 

An asynchronous system consists of the following sets: a set $P$ of process identifiers; a set $C$ of channels; a set $H_i$ is the set of possible local histories for each process $i$; a set $A$ of asynchronous runs; a set $M$ of all messages.

Each process / node in Lachesis selects $k$ other nodes as peers. 
For certain gossip protocol, nodes may be constrained to gossip with its $k$ peers. In such a case, the set of channels $C$ can be modelled as follows.
If node $i$ selects node $j$ as a peer, then $(i,j) \in C$. In general, one can express the history of each node in DAG-based protocol in general or in Lachesis protocol in particular, in the same manner as in the CCK paper~\cite{cck92}.


\dfnn{Consistent cut}{A consistent cut of a run $\sigma$ is any global state such that if $v_x^i \rightarrow v_y^j$ and $v_y^j$ is in the global state, then $v_x^i$ is also in the global state. Denoted by $\vec{c}(\sigma)$.}

The concept of consistent cut formalizes such a global state of a run. A consistent cut consists of all consistent DAG chains. A received event block exists in the global state implies the existence of the original event block. Note that a consistent cut is simply a vector of local states; we will use the notation $\vec{c}(\sigma)[i]$ to indicate the local state of $i$ in cut $\vec{c}$ of run $\sigma$.

A message chain of an asynchronous run is a sequence of messages $m_1$, $m_2$, $m_3$, $\dots$, such that, for all $i$, $receive(m_i)$ $\rightarrow$  $send(m_{i+1})$. Consequently, $send(m_1)$ $\rightarrow$ $receive(m_1)$ $\rightarrow$ $send(m_2)$ $\rightarrow$ $receive(m_2)$ $\rightarrow$ $send(m_3)$ $\dots$.

The formal semantics of an asynchronous system is given via  the satisfaction relation $\vdash$. Intuitively $\vec{c}(\sigma) \vdash \phi$, ``$\vec{c}(\sigma)$ satisfies $\phi$,'' if fact $\phi$ is true in cut $\vec{c}$ of run $\sigma$.

We assume that we are given a function $\pi$ that assigns a truth value to each primitive proposition $p$. The truth of a primitive proposition $p$ in $\vec{c}(\sigma)$ is determined by $\pi$ and $\vec{c}$. This defines $\vec{c}(\sigma) \vdash p$.\\

\dfnn{Equivalent cuts}{Two cuts $\vec{c}(\sigma)$ and $\vec{c'}(\sigma')$ are equivalent  with respect to $i$ if: $$\vec{c}(\sigma) \sim_i \vec{c'}(\sigma') \Leftrightarrow \vec{c}(\sigma)[i] = \vec{c'}(\sigma')[i]$$}

\dfnn{$i$ knows $\phi$}{$K_i(\phi)$ represents the statement ``$\phi$ is true in all possible consistent global states that include $i$’s local state''. 
$$\vec{c}(\sigma) \vdash K_i(\phi) \Leftrightarrow \forall \vec{c'}(\sigma')   (\vec{c'}(\sigma') \sim_i \vec{c}(\sigma) \ \Rightarrow\ \vec{c'}(\sigma') \vdash \phi) $$}

\dfnn{$i$ partially knows $\phi$}{$P_i(\phi)$ represents the statement ``there is some consistent global state in this run that includes $i$’s local state, in which $\phi$ is true.''
$$\vec{c}(\sigma) \vdash P_i(\phi) \Leftrightarrow \exists \vec{c'}(\sigma) ( \vec{c'}(\sigma) \sim_i \vec{c}(\sigma) \ \wedge\ \vec{c'}(\sigma) \vdash \phi )$$}
		
\dfnn{Majority concurrently knows}{The next modal operator is written $M^C$ and stands for ``majority concurrently knows.''
The definition of $M^C(\phi)$ is as follows.

$$M^C(\phi) =_{def} \bigwedge_{i \in S} K_i P_i(\phi), $$ where $S \subseteq P$ and $|S| > 2n/3$.}

This is adapted from the ``everyone concurrently knows'' in CCK paper~\cite{cck92}.
In the presence of one-third of faulty nodes, the original operator ``everyone concurrently knows'' is sometimes not feasible.
Our modal operator $M^C(\phi)$ fits precisely the semantics for BFT systems, in which unreliable processes may exist.\\

\dfnn{Concurrent common knowledge}{The last modal operator is concurrent common knowledge (CCK), denoted by $C^C$. $C^C(\phi)$ is defined as a fixed point of $M^C(\phi \wedge X)$.}

CCK defines a state of process knowledge that implies that all processes are in that same state of knowledge, with respect to $\phi$, along some cut of the run. In other words, we want a state of knowledge $X$ satisfying: $X = M^C(\phi \wedge X)$.	
$C^C$ will be defined semantically as the weakest such fixed point, namely as the greatest fixed-point of $M^C(\phi \wedge X)$.It therefore satisfies:

$$C^C(\phi) \Leftrightarrow  M^C(\phi \wedge C^C(\phi))$$

Thus, $P_i(\phi)$ states that there is some cut in the same asynchronous run $\sigma$ including $i$’s local state, such that $\phi$ is true in that cut.\\

Note that $\phi$ implies $P_i(\phi)$. But it is not the case, in general, that $P_i(\phi)$ implies $\phi$ or even that $M^C(\phi)$ implies $\phi$. The truth of $M^C(\phi)$ is determined with respect to some cut $\vec{c}(\sigma)$. A process cannot distinguish which cut, of the perhaps many cuts that are in the run and consistent with its local state, satisfies $\phi$; it can only know the existence of such a cut.\\ 

\dfnn{Global fact}{Fact $\phi$ is valid in system $\Sigma$, denoted by $\Sigma \vdash \phi$, if $\phi$ is true in all cuts of all runs of $\Sigma$.
	$$\Sigma \vdash \phi 
	\Leftrightarrow (\forall \sigma \in \Sigma)(\forall\vec{c}) (\vec{c}(a) \vdash \phi)$$}
	
Fact $\phi$ is valid, denoted $\vdash \phi$, if $\phi$ is valid in all systems, i.e. 
	$(\forall \Sigma) (\Sigma \vdash \phi)$.\\

\dfnn{Local fact}{A fact $\phi$ is local to process $i$ in system $\Sigma$ if
	 $\Sigma \vdash (\phi \Rightarrow K_i \phi)$.}

\subsection{Dominator (graph theory)}

In a graph $G=(V, E, r)$ a dominator is the relation between two vertices. A vertex $v$ is dominated by another vertex $w$, if every path in the graph from the root $r$ to $v$ have to go through $w$. Furthermore, the immediate dominator for a vertex $v$ is the last of $v$’s dominators, which every path in the graph have to go through to reach $v$.

\dfnn{Pseudo top}{A pseudo vertex, called top, is the parent of all top event blocks. Denoted by $\top$.}

\dfnn{Pseudo bottom}{A pseudo vertex, called bottom, is the child of all leaf event blocks. Denoted by $\bot$.}

With the pseudo vertices, we have $\bot$ happened-before all event blocks. Also all event blocks happened-before $\top$. That is, for all event $v_i$, $\bot \hbefore v_i$ and $v_i \hbefore \top$.

\dfnn{Dom}{An event $v_d$ dominates an event $v_x$ if every path from $\top$ to $v_x$ must go through $v_d$. Denoted by $v_d \dom v_x$.}

\dfnn{Strict dom}{An event $v_d$ strictly dominates an event $v_x$ if $v_d \dom v_x$ and $v_d$ does not equal $v_x$. Denoted by $v_d \sdom v_x$.}

\dfnn{Domfront}{A vertex $v_d$ is said ``domfront'' a vertex $v_x$ if  $v_d$ dominates an immediate predecessor of $v_x$, but $v_d$ does not strictly dominate $v_x$. Denoted by $v_d \domf v_x$.}

\dfnn{Dominance frontier}{The dominance frontier of a vertex $v_d$ is the set of all nodes $v_x$ such that $v_d \domf v_x$. Denoted by $DF(v_d)$.}

From the above definitions of domfront and dominance frontier, the following holds. If $v_d \domf v_x$, then $v_x \in DF(v_d)$.
	 
\subsection{OPERA chain (DAG)}

The core idea of the Lachesis protocol is to use a DAG-based structure, called the OPERA chain for our consensus algorithm. 
In the Lachesis protocol, a (participant) node is a server (machine) of the distributed system.
Each node can create messages, send messages to, and receive messages from, other nodes. The communication between nodes is asynchronous. 

Let $n$ be the number of participant nodes.
For consensus, the algorithm examines whether an event block is \emph{dominated} by $2n/3$ nodes, where $n$ is the number of all nodes. The Happen-before relation of event blocks with $2n/3$ nodes means that more than two-thirds of all nodes in the OPERA chain know the event block. 

The OPERA chain (DAG) is the local view of the DAG held by each node, this local view is used to identify topological ordering, select Clotho, and create time consensus through Atropos selection.
OPERA chain is a DAG graph $G = (V, E)$ consisting of $V$ vertices and $E$ edges. Each vertex $v_i \in V$ is an event block. An edge $(v_i,v_j) \in E$ refers to a hashing reference from $v_i$ to $v_j$; that is, $v_i \erefz v_j$.

\dfnn{Leaf}{The first created event block of a node is called a leaf event block.}

\dfnn{Root}{The leaf event block of a node is a root. When an event block $v$ can reach more than $2n/3$ of the roots in the previous frames, $v$ becomes a root.}

\dfnn{Root set}{The set of all first event blocks (leaf events) of all nodes form the first root set $R_1$ ($|R_1|$ = $n$). The root set $R_k$ consists of all roots $r_i$ such that $r_i$ $\not \in $ $R_i$, $\forall$ $i$ = 1..($k$-1) and $r_i$ can reach more than 2n/3 other roots in the current frame, $i$ = 1..($k$-1).}

\dfnn{Frame}{Frame $f_i$ is a natural number that separates Root sets. The root set at frame $f_i$ is denoted by $R_i$.}

\dfnn{Consistent chains}{OPERA chains $G_1$ and $G_2$ are consistent if for any event $v$ contained in both chains, $G_1[v] = G_2[v]$. Denoted by $G_1 \sim G_2$.}

When two consistent chains contain the same event $v$, both chains contain the same set of ancestors for $v$, with the same reference and self-ref edges between those ancestors.

If two nodes have OPERA chains containing event $v$, then they have the same $k$ hashes contained within $v$. A node will not accept an event during a sync unless that node already has $k$ references for that event, so both OPERA chains must contain $k$ references for $v$. The cryptographic hashes are assumed to be secure, therefore the references must be the same. By induction, all ancestors of $v$ must be the same. Therefore, the two OPERA chains are consistent.\\

\dfnn{Creator}{If a node $n_x$ creates an event block $v$, then the creator of $v$, denoted by $cr(v)$, is $n_x$.}

\dfnn{Consistent chain}{A global consistent chain $G^C$ is a chain if $G^C \sim G_i$ for all $G_i$.}

We denote $G \sqsubseteq G'$ to stand for $G$ is a subgraph of $G'$. Some properties of $G^C$ are given as follows:

\begin{enumerate}[(1)]
	\item $\forall G_i$ ($G^C \sqsubseteq G_i$).
	\item
	$\forall v \in G^C$ $\forall G_i$ ($G^C[v] \sqsubseteq G_i[v]$).
	\item
	($\forall v_c \in G^C$) ($\forall v_p \in G_i$) (($v_p \hbefore v_c) \Rightarrow v_p \in G^C$).
\end{enumerate}

\dfnn{Consistent root}{Two chains $G_1$ and $G_2$ are root consistent, if for every $v$ contained in both chains, $v$ is a root of $j$-th frame in $G_1$, then $v$ is a root of $j$-th frame in $G_2$.}

By consistent chains, if $G_1 \sim G_2$ and $v$ belongs to both chains, then $G_1[v]$ = $G_2[v]$.
We can prove the proposition by induction. For $j$ = 0, the first root set is the same in both $G_1$ and $G_2$. Hence, it holds for $j$ = 0. Suppose that the proposition holds for every $j$ from 0 to $k$. We prove that it also holds for $j$= $k$ + 1.
Suppose that $v$ is a root of frame $f_{k+1}$ in $G_1$. 
Then there exists a set $S$ reaching 2/3 of members in $G_1$ of frame $f_k$ such that $\forall u \in S$ ($u\hbefore v$). As $G_1 \sim G_2$, and $v$ in $G_2$, then $\forall u \in S$ ($u \in G_2$). Since the proposition holds for $j$=$k$, 
As $u$ is a root of frame $f_{k}$ in $G_1$, $u$ is a root of frame $f_k$ in $G_2$. Hence, the set $S$ of 2/3 members $u$ happens before $v$ in $G_2$. So $v$ belongs to $f_{k+1}$ in $G_2$.
	
Thus, all nodes have the same consistent root sets, which are the root sets in $G^C$. Frame numbers are consistent for all nodes.\\

\dfnn{Flag table}{A flag table stores reachability from an event block to another root. The sum of all reachabilities, namely all values in flag table, indicates the number of reachabilities from an event block to other roots.}

\dfnn{Consistent flag table}{For any top event $v$ in both OPERA chains $G_1$ and $G_2$, and $G_1 \sim G_2$, then the flag tables of $v$ are consistent if they are the same in both chains.}

From the above, the root sets of $G_1$ and $G_2$ are consistent. If $v$ contained in $G_1$, and $v$ is a root of $j$-th frame in $G_1$, then $v$ is a root of $j$-th frame in $G_i$. Since $G_1 \sim G_2$, $G_1[v] = G_2[v]$. The reference event blocks of $v$ are the same in both chains. Thus the flag tables of $v$ of both chains are the same.\\

\dfnn{Clotho}{A root $r_k$ in the frame $f_{a+3}$ can nominate a root $r_a$ as Clotho if more than 2n/3 roots in the frame $f_{a+1}$ dominate $r_a$ and $r_k$ dominates the roots in the frame $f_{a+1}$.}

Each node nominates a root into Clotho via the flag table. If all nodes have an OPERA chain with same shape, the values in flag table will be equal to each other in OPERA chain. Thus, all nodes nominate the same root into Clotho since the OPERA chain of all nodes has same shape.

\dfnn{Atropos}{An Atropos is assigned consensus time through the Lachesis consensus algorithm and is utilized for determining the order between event blocks. Atropos blocks form a Main-chain, which allows time consensus ordering and responses to attacks.}

For any root set $R$ in the frame $f_{i}$, the time consensus algorithm checks whether more than 2n/3 roots in the frame $f_{i-1}$ selects the same value. However, each node selects one of the values collected from the root set in the previous frame by the time consensus algorithm and Reselection process. Based on the Reselection process, the time consensus algorithm can reach agreement. However, there is a possibility that consensus time candidate does not reach agreement~\cite{Fischer85}. To solve this problem, time consensus algorithm includes minimal selection frame per next $h$ frame. In minimal value selection algorithm, each root selects minimum value among values collected from previous root set. Thus, the consensus time reaches consensus by time consensus algorithm.

\dfnn{Main-chain (Blockchain)}{For faster consensus, \emph{Main-chain} is a special sub-graph of the OPERA chain (DAG).}

The Main chain --- a core subgraph of OPERA chain, plays the important role of ordering the event blocks. The Main chain stores shortcuts to connect between the Atropos. 
After the topological ordering is computed over all event blocks through the Lachesis protocol, Atropos blocks are determined and form the Main chain.  To improve path searching, we use a flag table --- a local hash table structure as a cache that is used to quickly determine the closest root to a event block.

In the OPERA chain, an event block is called a \emph{root} if the event block is linked to more than two-thirds of previous roots. A leaf vertex is also a root itself. With root event blocks, we can keep track of ``vital'' blocks that $2n/3$ of the network agree on.

\begin{figure} \centering  
\includegraphics[height=8cm, width=1.0\columnwidth]{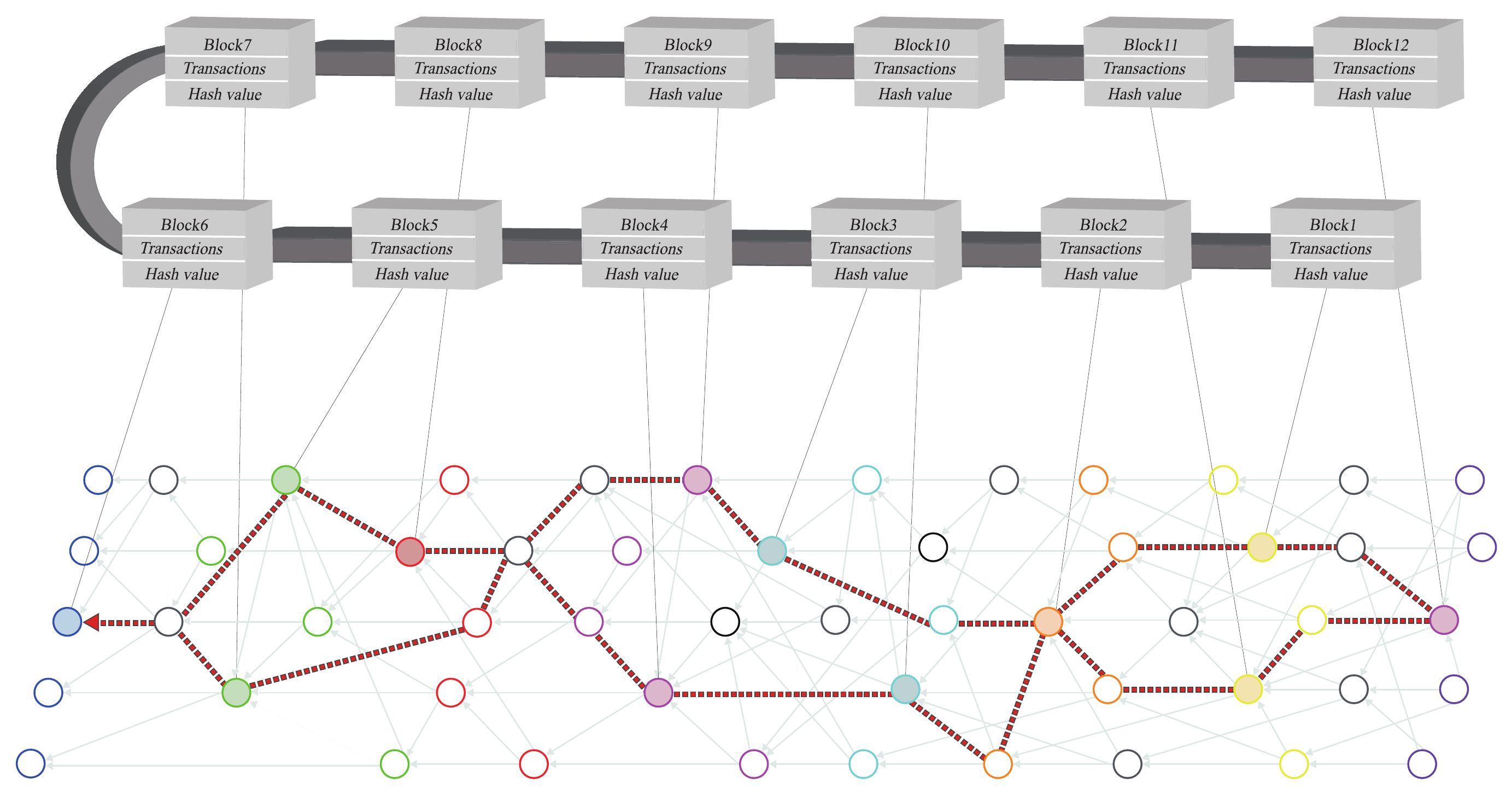}
\caption{An Example of Main-chain}
\label{fig:mainchain}
\end{figure}

Figure~\ref{fig:mainchain} shows an example of the Main chain composed of Atropos event blocks. 
In particular, the Main chain consists of Atropos blocks that are derived from root blocks and so are agreed by $2n/3$ of the network nodes. Thus, this guarantees that at least $2n/3$ of nodes have come to consensus on this Main chain. 

Each participant node has a copy of the Main chain and can search consensus position of its own event blocks.
Each event block can compute its own consensus position by checking the nearest Atropos event block. Assigning and searching consensus position are introduced in the consensus time selection section. 

The Main chain provides quick access to the previous transaction history to efficiently process new incoming event blocks. From the Main chain, information about unknown participants or attackers can be easily viewed.
The Main chain can be used efficiently in transaction information management by providing quick access to new event blocks that have been agreed on by the majority of nodes. In short, the Main-chain gives the following advantages:

-	All event blocks or nodes do not need to store all information. It is efficient for data management.

-	Access to previous information is efficient and fast.

Based on these advantages, OPERA chain can respond strongly to efficient transaction treatment and attacks through its Main-chain.

\newpage
\section{Lachesis Protocol}\label{se:lca}

\begin{algorithm}[H]
\caption{Main Procedure}\label{al:main}
\begin{algorithmic}[1]
	\Procedure{Main Procedure}{}
	\BState \emph{loop}:
	\State A, B = $k$-node Selection algorithm()
	\State Request sync to node A and B
	\State Sync all known events by Lachesis protocol
	\State Event block creation
	\State (optional) Broadcast out the message
	\State Root selection
	\State Clotho selection
	\State Atropos time consensus
	\BState \emph{loop}:
	\State Request sync from a node
	\State Sync all known events by Lachesis protocol
	\EndProcedure
\end{algorithmic}
\end{algorithm}

Algorithm~\ref{al:main} shows the pseudo algorithm for the Lachesis core procedure. The algorithm consists of two parts and runs them in parallel.

- In part one, each node requests synchronization and creates event blocks. In line 3, a node runs the Node Selection Algorithm. The Node Selection Algorithm returns the $k$ IDs of other nodes to communicate with. In line 4 and 5, the node synchronizes the OPERA chain (DAG) with the other nodes. Line 6 runs the Event block creation, at which step the node creates an event block and checks whether it is a root. The node then broadcasts the created event block to all other known nodes in line 7. The step in this line is optional. In line 8 and 9, Clotho selection and Atropos time consensus algorithms are invoked. The algorithms determines whether the specified root can be a Clotho, assign the consensus time, and then confirm the Atropos. 

- The second part is to respond to synchronization requests. In line 10 and 11, the node receives a synchronization request and then sends its response about the OPERA chain.

\subsection{Peer selection algorithm}
In order to create an event block, a node needs to select $k$ other nodes. Lachesis protocols does not depend on how peer nodes are selected. One simple approach can use a random selection from the pool of $n$ nodes. The other approach is to define some criteria or cost function to select other peers of a node. 

Within distributed system, a node can select other nodes with low communication costs, low network latency, high bandwidth, high successful transaction throughputs.

\subsection{Dynamic participants}
Our Lachesis protocol allows an arbitrary number of participants to dynamically join the system. The OPERA chain (DAG) can still operate with new participants. Computation on flag tables is set based and independent of which and how many participants have joined the system. Algorithms for selection of Roots, Clothos and Atroposes are flexible enough and not dependence on a fixed number of participants.

\newpage
\subsection{Peer synchronization}

We describe an algorithm that synchronizes events between the nodes.

\begin{algorithm}[H]
	\caption{EventSync}\label{al:syncevents}
	\begin{algorithmic}[1]
		\Procedure{sync-events()}{}		
		\State Node $n_1$ selects random peer to synchronize with
		\State $n_1$ gets local known events (map[int]int)
		\State $n_1$ sends RPC request Sync request to peer
		\State $n_2$ receives RPC requestSync request
		\State $n_2$ does an EventDiff check on the known map (map[int]int)
		
		\State $n_2$ returns unknown events, and map[int]int of known events to $n_1$		
		
		\EndProcedure
	\end{algorithmic}
\end{algorithm}

The algorithm assumes that a node always needs the events in topological ordering (specifically in reference to the lamport timestamps), an alternative would be to use an inverse bloom lookup table (IBLT) for completely potential randomized events. \\

Alternatively, one can simply use a fixed incrementing index to keep track of the top event for each node.\\

\subsection{Node Structure}
This section gives an overview of the node structure in Lachesis.

Each node has a height vector, in-degree vector, flag table, frames, clotho check list, max-min value, main-chain (blockchain), and their own local view of the OPERA chain (DAG). The height vector is the number of event blocks created by the $i$-th node. The in-degree vector refers to the number of edges from other event blocks created by other nodes to the top event block of this node. The top event block indicates the most recently created event block by this node. The flag table is a $n x k$ matrix, where n is the number of nodes and k is the number of roots that an event block can reach. If an event block $e$ created by $i$-th node can reach $j$-th root, then the flag table stores the hash value of the $j$-th root. Each node maintains the flag table of each top event block. 

Frames store the root set in each frame. Clotho check list has two types of check points; Clotho candidate ($CC$) and Clotho ($C$). If a root in a frame is a $CC$, a node check the $CC$ part and if a root becomes Clotho, a node check $C$ part. Max-min value is timestamp that addresses for Atropos selection. The Main-chain is a data structure storing hash values of the Atropos blocks. 

\begin{figure}[H] \centering  
\includegraphics[width=.5\textwidth]{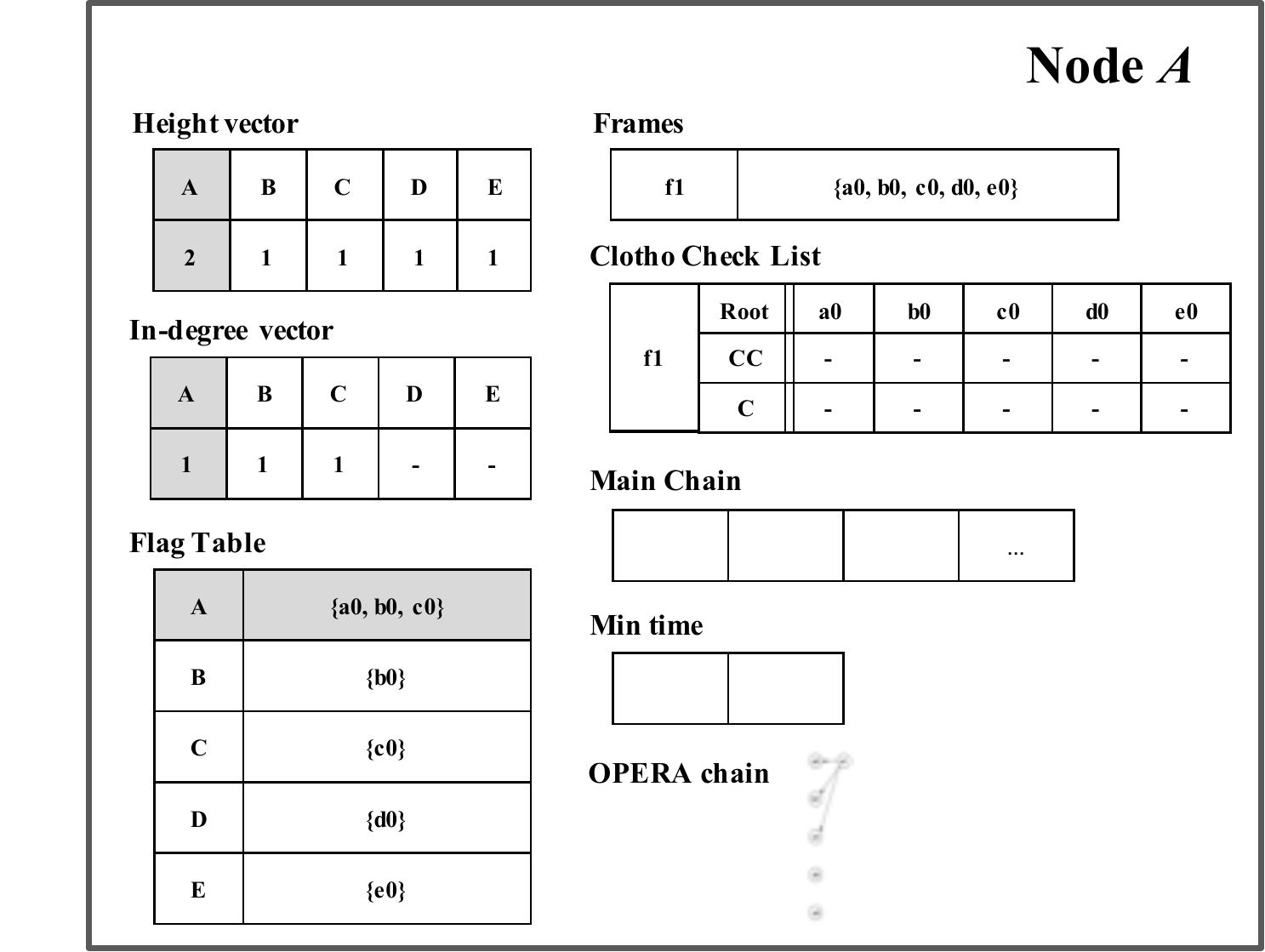}
\caption{An Example of Node Structure}
\label{fig:node}
\end{figure}

Figure~\ref{fig:node} shows an example of the node structure component of a node $A$. In the figure, each value excluding self height in the height vector is 1 since the initial state is shared to all nodes. In the in-degree vector, node $A$ stores the number of edges from other event blocks created by other nodes to the top event block. The in-degrees of node $A$, $B$, and $C$ are 1. 
In flag table, node $A$ knows other two root hashes since the top event block can reach those two roots. Node $A$ also knows that other nodes know their own roots. In the example situation there is no clotho candidate and Clotho, and thus clotho check list is empty. The main-chain and max-min value are empty for the same reason as clotho check list. 

\subsection{Peer selection algorithm via Cost function}
We define three versions of the Cost Function ($C_F$). Version one is focused around updated information share and is discussed below. The other two versions are focused on root creation and consensus facilitation, these will be discussed in a following paper.

We define a Cost Function ($C_F$) for preventing the creation of lazy nodes. A lazy node is a node that has a lower work portion in the OPERA chain (has created fewer event blocks). When a node creates an event block, the node selects other nodes with low value outputs from the cost function and refers to the top event blocks of the reference nodes. An equation~(\ref{eq1}) of $C_F$ is as follows,

\begin{equation}\label{eq1}
C_{F} =I/H
\end{equation}

where $I$ and $H$ denote values of in-degree vector and height vector respectively. If the number of nodes with the lowest $C_F$ is more than $k$, one of the nodes is selected at random. The reason for selecting high $H$ is that we can expect a high possibility to create a root because the high $H$ indicates that the communication frequency of the node had more opportunities than others with low $H$. Otherwise, the nodes that have high $C_F$ (the case of $I$ $>$ $H$) have generated fewer event blocks than the nodes that have low $C_F$.. 

\begin{figure}[H] \centering  
\includegraphics[width=.8\textwidth]{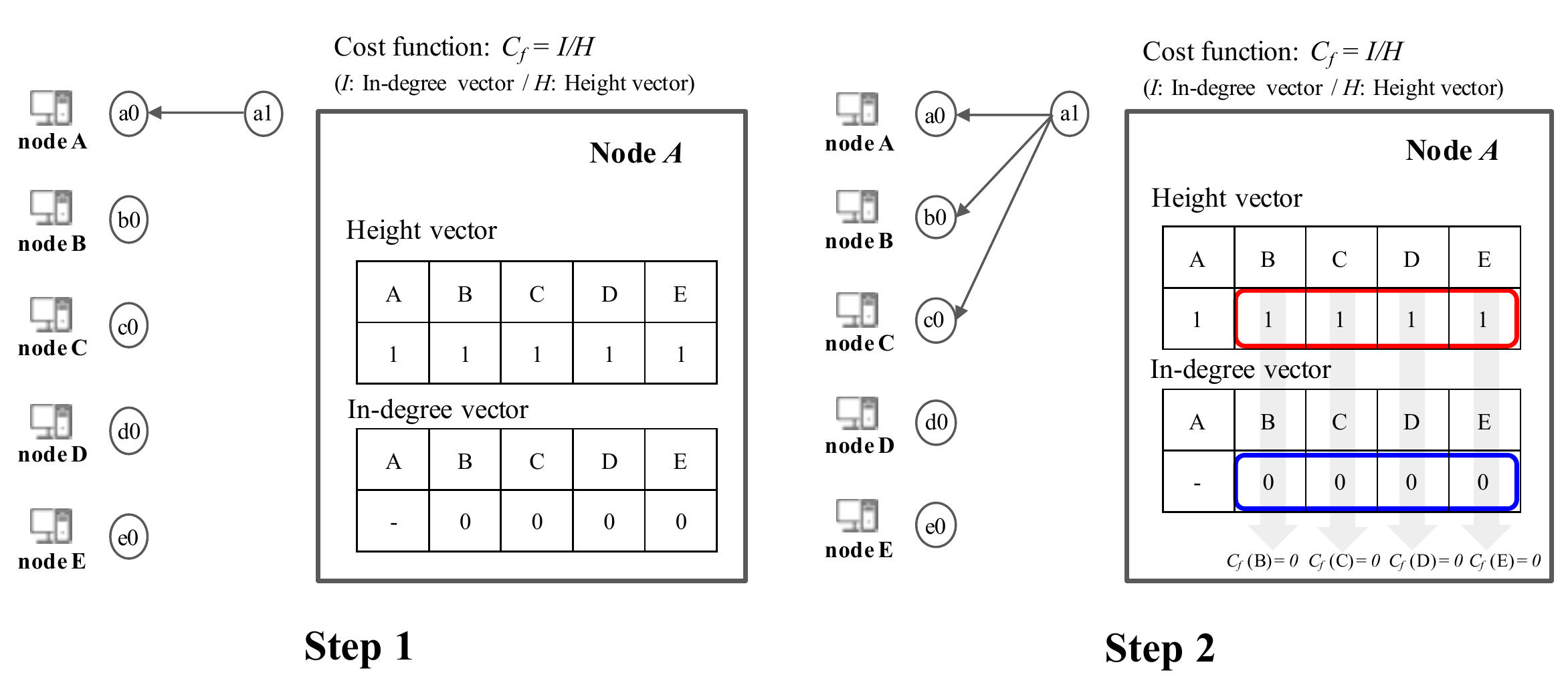}
\caption{An Example of Cost Function 1}
\label{fig:costfunction_1}
\end{figure}

Figure~\ref{fig:costfunction_1} shows an example of the node selection based on the cost function after the creation of leaf events by all nodes. In this example, there are five nodes and each node created leaf events. All nodes know other leaf events. Node $A$ creates an event block $v_1$ and $A$ calculates the cost functions. Step 2 in Figure~\ref{fig:costfunction_1} shows the results of cost functions based on the height and in-degree vectors of node $A$. In the initial step, each value in the vectors are same because all nodes have only leaf events. Node $A$ randomly selects $k$ nodes and connects $v_1$ to the leaf events of selected nodes. In this example, we set $k$=3 and assume that node $A$ selects node $B$ and $C$. 

\begin{figure}[H] \centering  
\includegraphics[width=.8\textwidth]{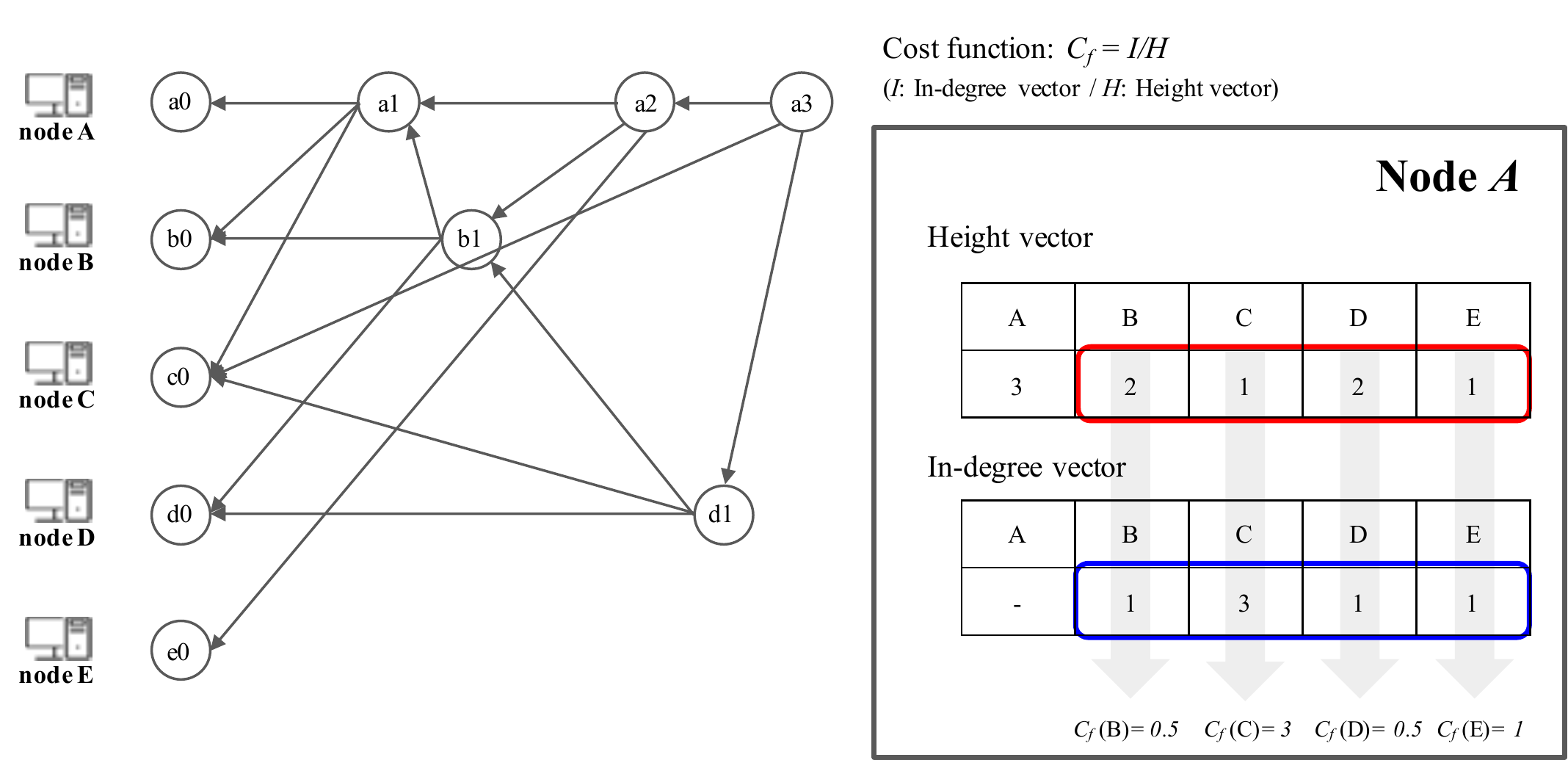}
\caption{An Example of Cost Function 2}
\label{fig:costfunction_2}
\end{figure}

Figure~\ref{fig:costfunction_2} shows an example of the node selection after a few steps of the simulation in Figure~\ref{fig:costfunction_1}. In Figure~\ref{fig:costfunction_2}, the recent event block is $v_5$ created by node $A$. Node $A$ calculates the cost function and selects the other two nodes that have the lowest results of the cost function. In this example, node $B$ has 0.5 as the result and other nodes have the same values. Because of this, node $A$ first selects node $B$ and randomly selects other nodes among nodes $C$, $D$, and $E$.

The height of node $D$ in the example is 2 (leaf event and event block $v_4$). On the other hand, the height of node $D$ in node structure of $A$ is 1. Node $A$ is still not aware of the presence of the event block $v_4$. It means that there is no path from the event blocks created by node $A$ to the event block $v_4$. Thus, node $A$ has 1 as the height of node $D$. 

\begin{algorithm}[H]
\caption{$k$-neighbor Node Selection}\label{al:ns}
\begin{algorithmic}[1]
	\Procedure{$k$-node Selection}{}
	\State \textbf{Input:} Height Vector $H$, In-degree Vector $I$
	\State \textbf{Output:} reference node $ref$
	\State min\_cost $\leftarrow$ $INF$ 
	\State $s_{ref}$ $\leftarrow$ None
	\For{$k \in Node\_Set$}
	\State $c_f$ $\leftarrow$ $\frac{I_k}{H_k}$ 
	\If {min\_cost $>$ $c_f$} 
	\State min\_cost $\leftarrow$ $c_f$
	\State $s_{ref}$ $\leftarrow$ {k}
	\ElsIf {min\_cost $equa$l $c_f$}
	\State $s_{ref}$ $\leftarrow$ $s_{ref}$ $\cup$ $k$
	\EndIf
	\EndFor
	\State $ref$ $\leftarrow$ random select in $s_{ref}$
	\EndProcedure
\end{algorithmic}
\end{algorithm}

Algorithm~\ref{al:ns} shows the selecting algorithm for selecting reference nodes. The algorithm operates for each node to select a communication partner from other nodes. Line 4 and 5 set min\_cost and $S_{ref}$ to initial state. Line 7 calculates the cost function $c_f$ for each node. In line 8, 9, and 10, we find the minimum value of the cost function and set min\_cost and $S_{ref}$ to $c_f$ and the ID of each node respectively. Line 11 and 12 append the ID of each node to $S_{ref}$ if min\_cost equals $c_f$. Finally, line 13 selects randomly $k$ node IDs from $S_{ref}$ as communication partners. The time complexity of Algorithm 2 is $O(n)$, where $n$ is the number of nodes. 

After the reference node is selected, each node communicates and shares information of all event blocks known by them. A node creates an event block by referring to the top event block of the reference node. The Lachesis protocol works and communicates asynchronously. This allows a node to create an event block asynchronously even when another node creates an event block. The communication between nodes does not allow simultaneous communication with the same node. 

\begin{figure}[H] \centering  
\includegraphics[height=7cm, width=1.0\columnwidth]{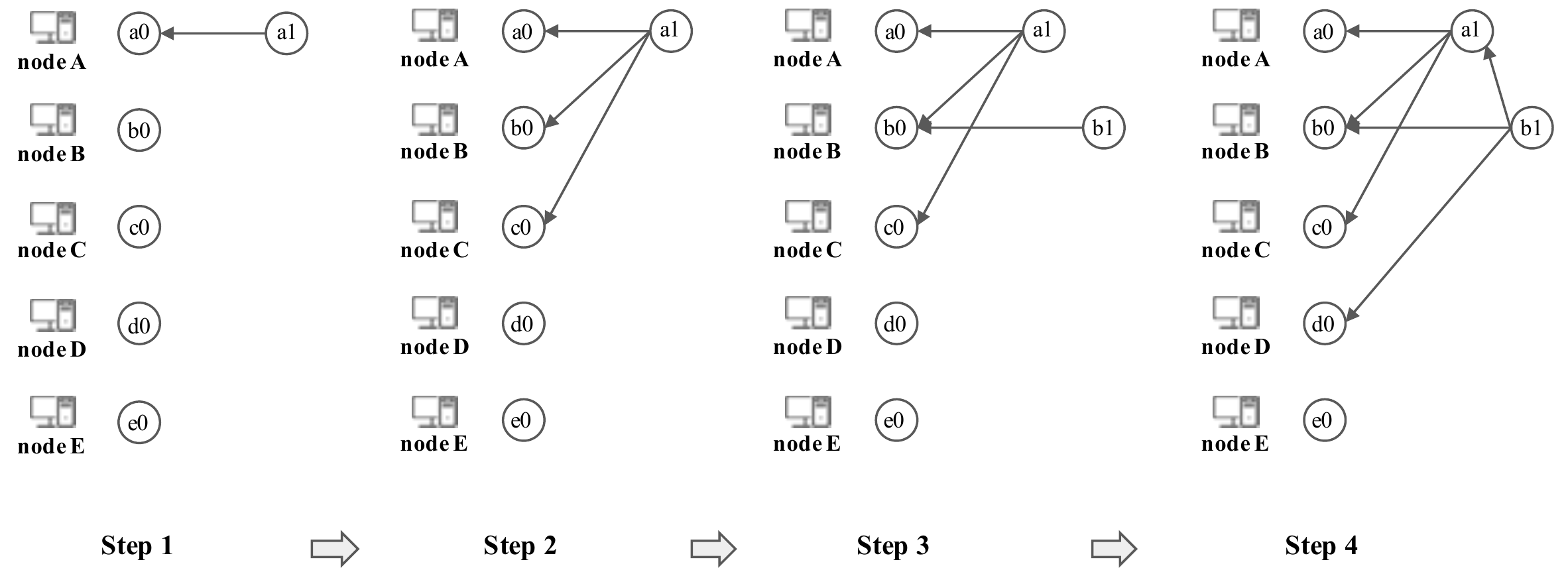}
\caption{An Example of Node Selection}
\label{fig:node_selection}
\end{figure}

Figure~\ref{fig:node_selection} shows an example of the node selection in Lachesis protocol. In this example, there are five nodes ($A, B, C, D,$ and $E$) and each node generates the first event blocks, called leaf events. All nodes share other leaf events with each other. In the first step, node $A$ generates new event block $a_1$. Then node $A$ calculates the cost function to connect other nodes. In this initial situation, all nodes have one event block called leaf event, thus the height vector and the in-degree vector in node $A$ has same values. In other words, the heights of each node are 1 and in-degrees are 0. Node $A$ randomly select the other two nodes and connects $a_1$ to the top two event blocks from the other two nodes. Step 2 shows the situation after connections. In this example, node $A$ select node $B$ and $C$ to connect $a_1$ and the event block $a_1$ is connected to the top event blocks of node $B$ and $C$. Node $A$ only knows the situation of the step 2. 

After that, in the example, node $B$ generates a new event block $b_1$ and also calculates the cost function. $B$ randomly select the other two nodes; $A$, and $D$, since $B$ only has information of the leaf events. Node $B$ requests to $A$ and $D$ to connect $b_1$, then nodes $A$ and $D$ send information for their top event blocks to node $B$ as response. The top event block of node $A$ is $a_1$ and node $D$ is the leaf event. The event block $b_1$ is connected to $a_1$ and leaf event from node $D$. Step 4 shows these connections. 

\subsection{Event block creation}
In the Lachesis protocol, every node can create an event block. Each event block refers to other $k$ event blocks using their hash values. In the Lachesis protocol, a new event block refers to $k$-neighbor event blocks under the following conditions:

\begin{enumerate}
\item Each of the $k$ reference event blocks is the top event blocks of its own node.
\item One reference should be made to a self-ref that references to an event block of the same node. 
\item The other $k$-1 reference refers to the other $k$-1 top event nodes on other nodes.
\end{enumerate}

\newpage

\begin{figure}[H] \centering  
\includegraphics[width=1.0\textwidth]{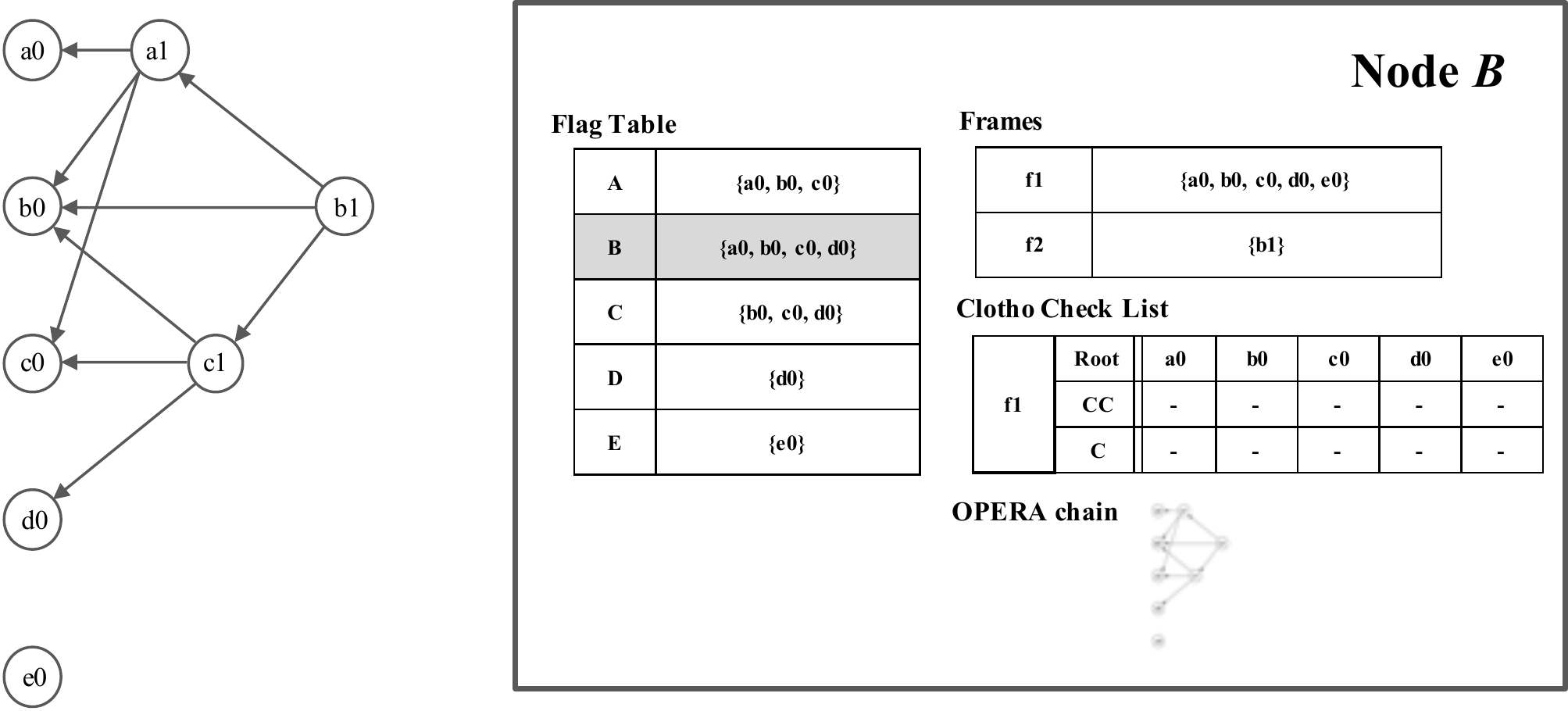}
\caption{An Example of Event Block Creation with Flag Table}
\label{fig:ex_ebc}
\end{figure}

Figure~\ref{fig:ex_ebc} shows the example of an event block creation with a flag table. In this example the recent created event block is $b_1$ by node $B$. The figure shows the node structure of node $B$. We omit the other information such as height and in-degree vectors since we only focus on the change of the flag table with the event block creation in this example. The flag table of $b_1$ in Figure~\ref{fig:ex_ebc} is updated with the information of the previous connected event blocks $a_1$, $b_0$, and $c_1$. Thus, the set of the flag table is the 
results of OR operation among the three root sets for $a1$ ($a_0$, $b_0$, and $c_0$), $b_0$ ($b_0$), and $c_1$ ($b_0$, $c_0$, and $d_0$). 

\begin{figure}[H] \centering  
\includegraphics[height=6cm, width=0.7\columnwidth]{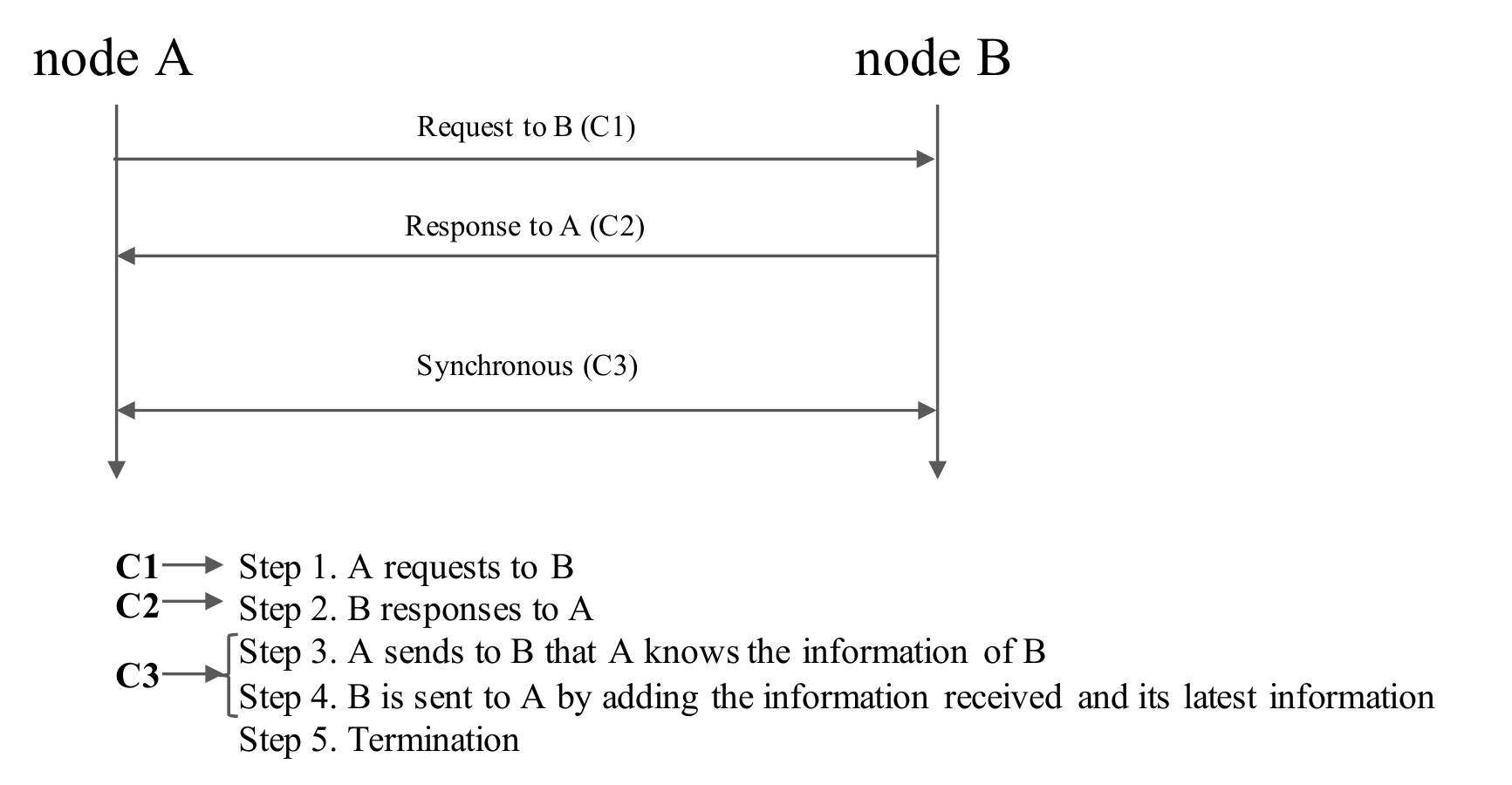}
\caption{An Example of Communication Process}
\label{fig:communication process}
\end{figure}

Figure~\ref{fig:communication process}, shows the communication process is divided into five steps for two nodes to create an event block. Simply, a node $A$ requests to $B$. then, $B$ responds to $A$ directly.

\subsection{Topological ordering of events using Lamport timestamps}
Every node has a physical clock and it needs physical time to create an event block. However, for consensus, Lachesis protocols relies on a logical clock for each node. For the purpose, we use \textit{"Lamport timestamps"} \cite{lamport1978time} to determine the time ordering between event blocks in a asynchronous distributed system.

\begin{figure}[H] \centering  
	\includegraphics[width=0.7\textwidth, width=1.0\columnwidth]{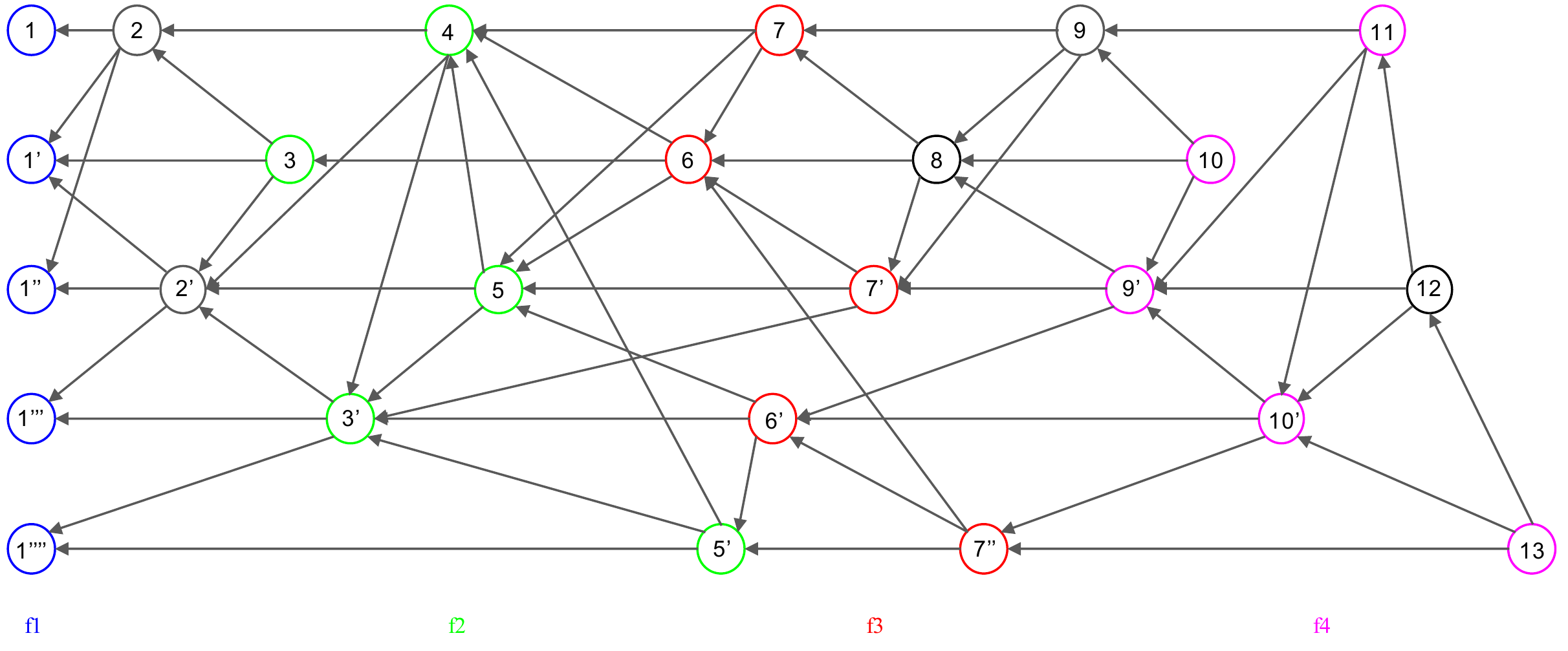}
	\caption{An example of Lamport timestamps}
	\label{fig:Lamport}
\end{figure}

The Lamport timestamps algorithm is as follows:

\begin{enumerate}
\item Each node increments its count value before creating an event block.
\item When sending a message include its count value, receiver should consider which sender’s message is received and increments its count value.
\item If current counter is less than or equal to the received count value from another node, then the count value of the recipient is updated.  
\item If current counter is greater than the received count value from another node, then the current count value is updated.
\end{enumerate}

We use the Lamport's algorithm to enforce a topological ordering of event blocks and use it in the Atropos selection algorithm. 

Since an event block is created based on logical time, the sequence between each event blocks is immediately determined. Because the Lamport timestamps algorithm gives a partial order of all events, the whole time ordering process can be used for Byzantine fault tolerance.

\subsection{Domination Relation}

Here, we introduce a new idea that extends the concept of domination.\\

For a vertex $v$ in a DAG $G$, let $G[v] = (V_v,E_v)$ denote an induced-subgraph of $G$ such that $V_v$ consists of all ancestors of $v$ including $v$, and $E_v$ is the induced edges of $V_v$ in $G$.\\

For a set $S$ of vertices, an event $v_d$  $\frac{2}{3}$-dominates $S$ if there are more than 2/3 of vertices $v_x$ in $S$ such that $v_d$ dominates $v_x$. 	
Recall that $R_1$ is the set of all leaf vertices in $G$. The $\frac{2}{3}$-dom set $D_0$ is the same as the set $R_1$.The $\frac{2}{3}$-dom set $D_i$ is defined as follows:	\\

A vertex $v_d$ belongs to a $\frac{2}{3}$-dom set within the graph $G[v_d]$, if $v_d$ $\frac{2}{3}$-dominates $R_1$.
	The $\frac{2}{3}$-dom set $D_k$ consists of all roots $d_i$ such that  $d_i$ $\not \in $ $D_i$, $\forall$ $i$ = 1..($k$-1), and $d_i$ $\frac{2}{3}$-dominates $D_{i-1}$.\\
	
The $\frac{2}{3}$-dom set $D_i$ is the same with the root set $R_i$, for all nodes.\\

\subsection{Examples of domination relation in DAGs}

This section gives several examples of DAGs and the domination relation between their event blocks.
\begin{figure}[h]
	\centering
	(a)\includegraphics[width=0.95\linewidth]{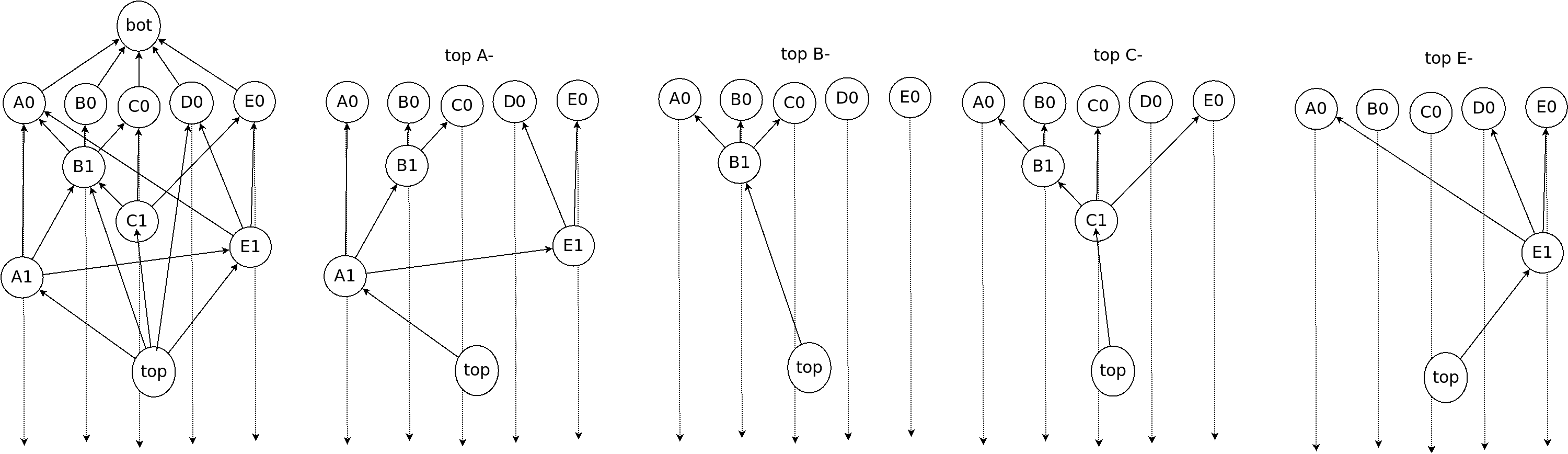}\\
	(b)\includegraphics[width=0.95\linewidth]{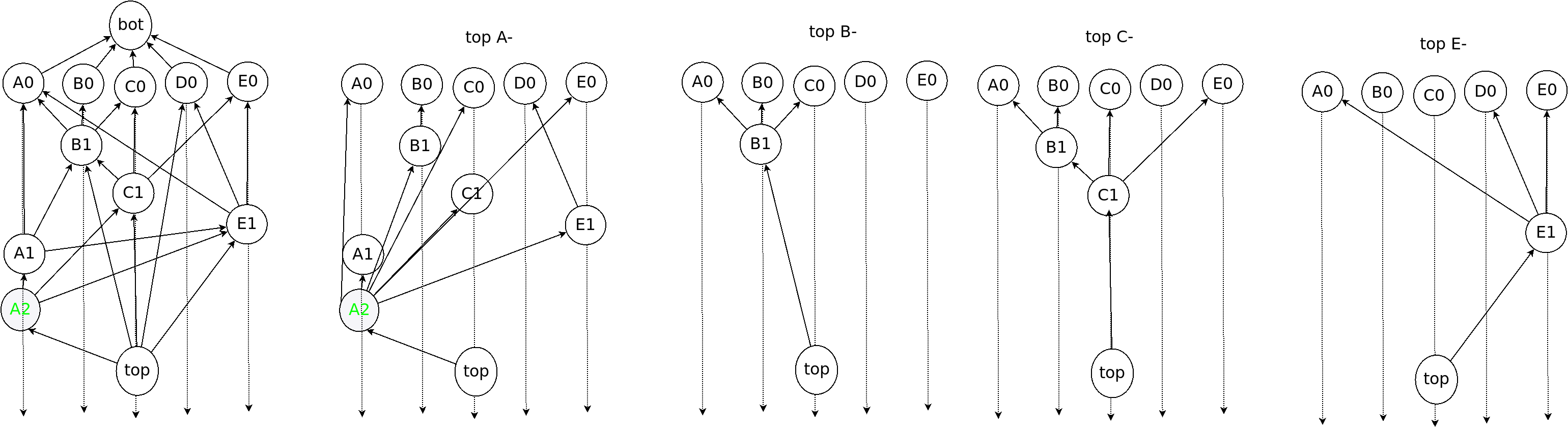}
	\caption{Examples of OPERA chain and dominator tree}
	\label{fig:domtrees1}
\end{figure}

Figure~\ref{fig:domtrees1} shows an examples of a DAG and dominator trees. 

\begin{figure}[H]
	\centering
	\includegraphics[width=0.5\linewidth]{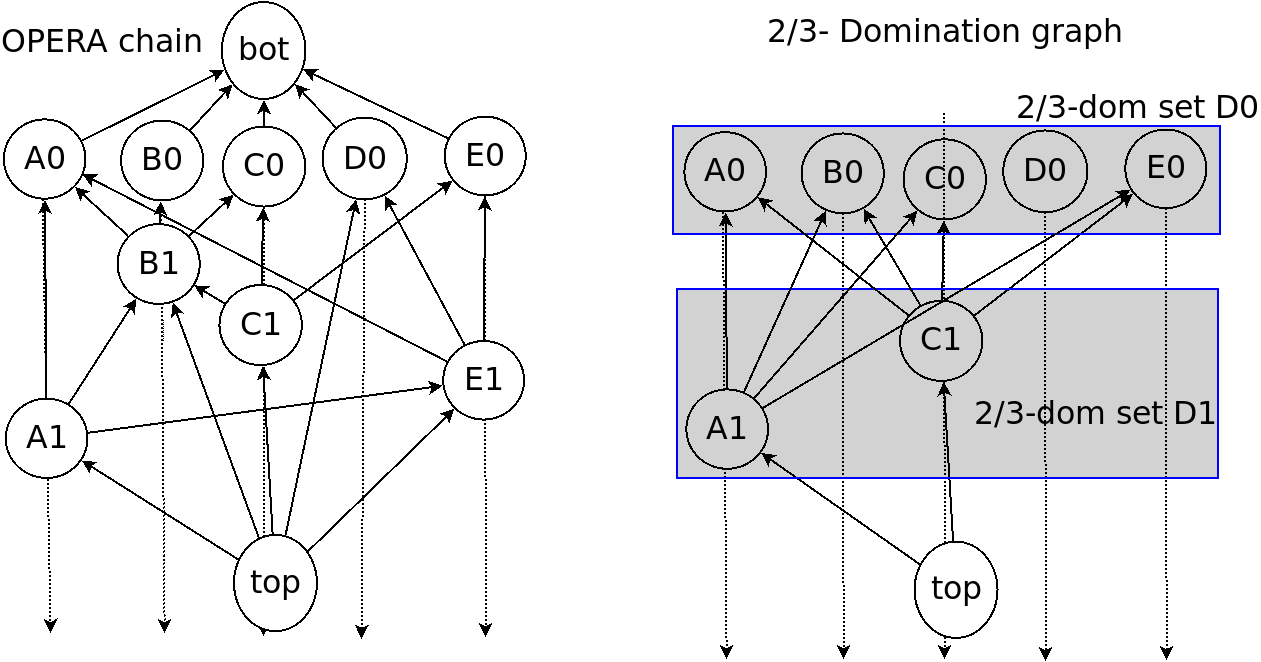}
	\caption{An example of OPERA chain and its 2/3 domination graph. The $\frac{2}{3}$-dom sets are shown in grey.}
	\label{fig:domset1}
\end{figure}

Figure~\ref{fig:domset1} depicts an example of a DAG and 2/3 dom sets.

\begin{figure}[H]
	\centering
	(a) \includegraphics[width=0.9\linewidth]{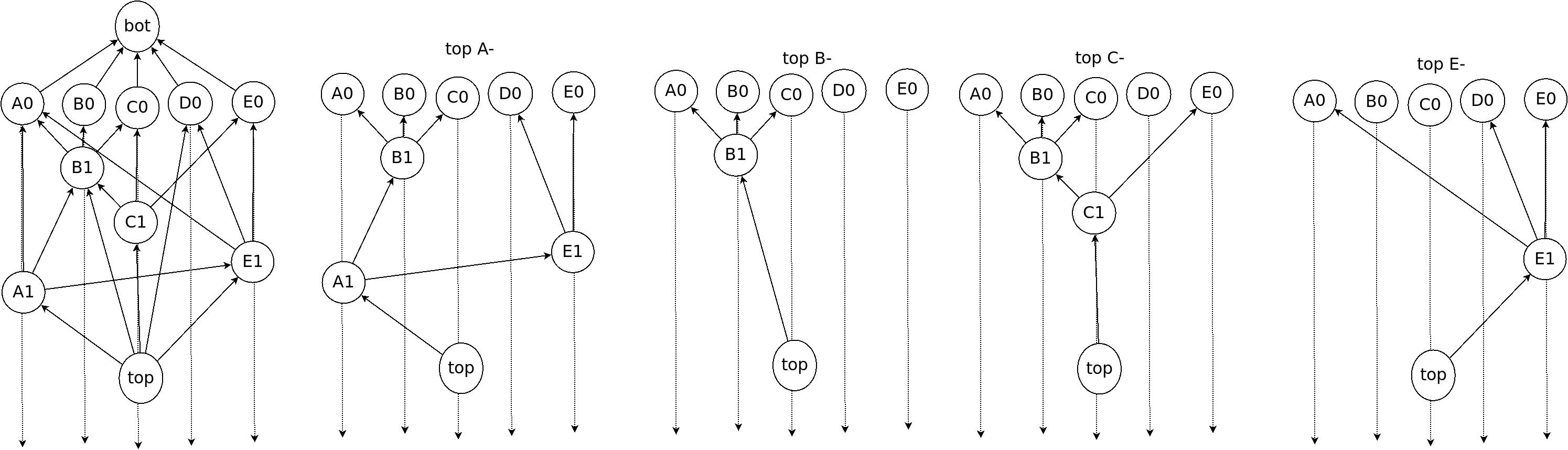}  \\
	(b) \includegraphics[width=0.9\linewidth]{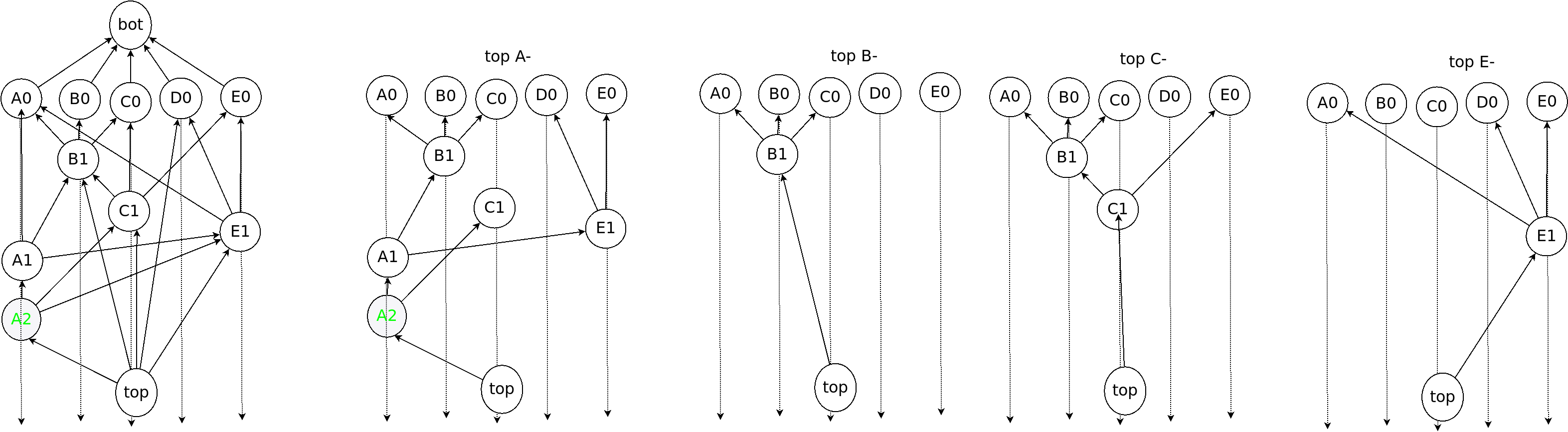} \\
	(c) \includegraphics[width=0.9\linewidth]{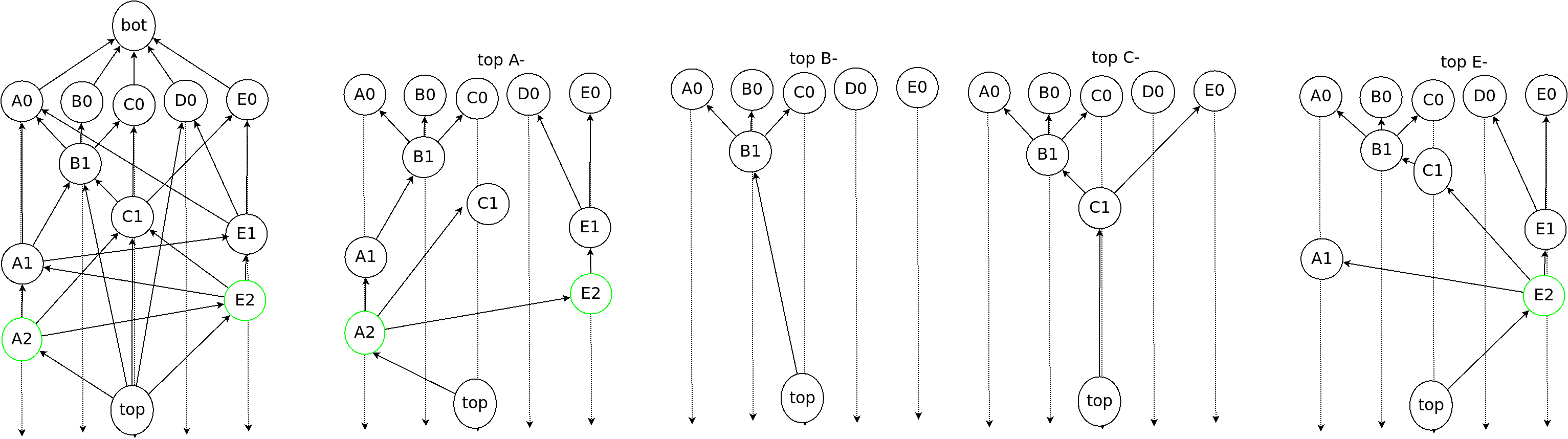}
	\caption{An example of dependency graphs on individual nodes. From (a)-(c) there is one new event block appended. There is no fork, the simplified dependency graphs become trees.}
	\label{fig:deptreesmod1}
\end{figure}

Figure~\ref{fig:deptreesmod1} shows an example an dependency graphs. On each row, the left most figure shows the latest OPERA chain. The left figures on each row depict the dependency graphs of each node, which are in their compact form. When no fork presents, each of the compact dependency graphs is a tree.

\newpage

\begin{figure}[H]	\centering
	(a)\includegraphics[width=0.9\linewidth]{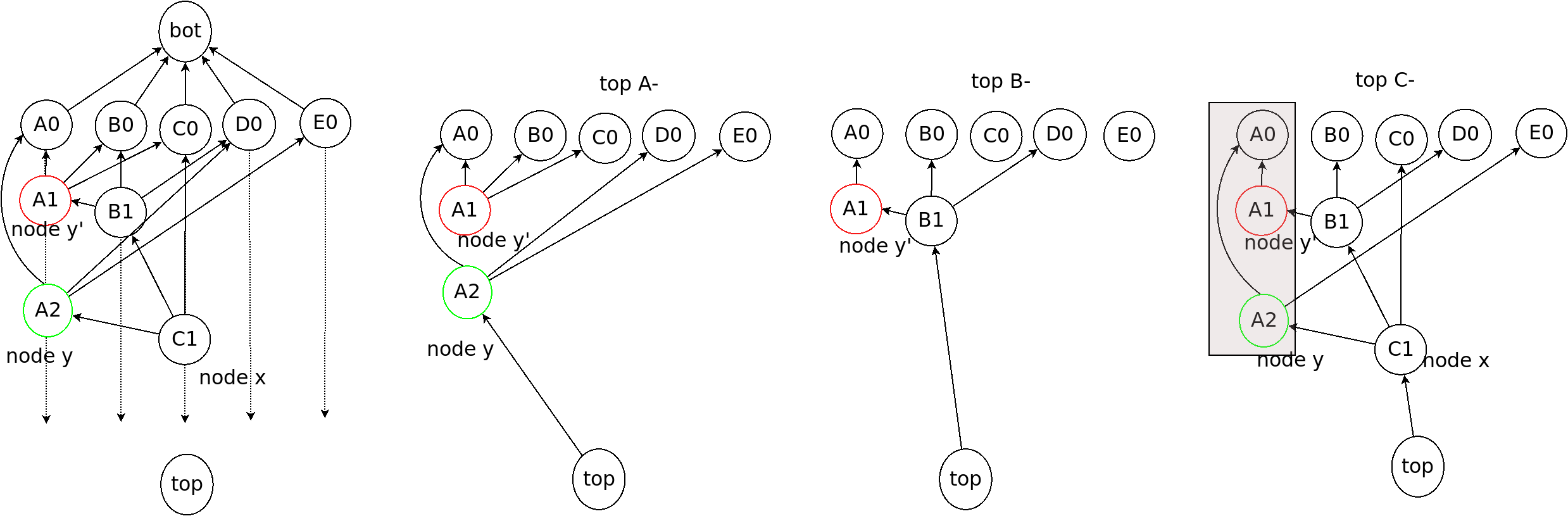}\\	(b)\includegraphics[width=0.9\linewidth]{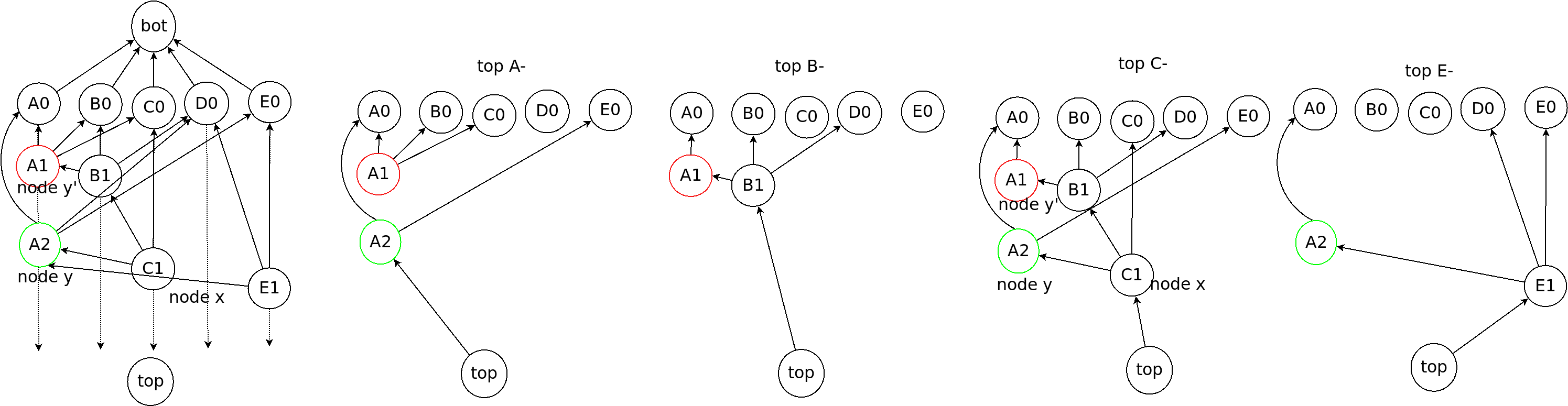}\\
	(c)\includegraphics[width=0.9\linewidth]{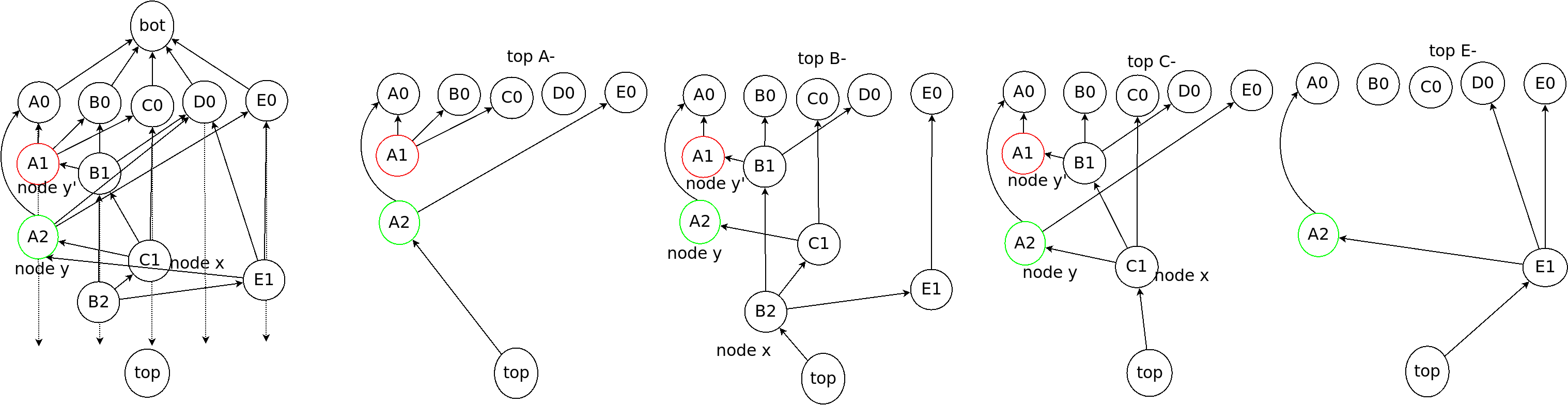}
	\caption{An example of a pair of fork events in an OPERA chain. The fork events are shown in red and green. The OPERA chains from (a) to (d) are different by adding one single event at a time.}
	\label{fig:deptreesfork1}
\end{figure}

Figure~\ref{fig:deptreesfork1} shows an example of a pair of fork events. Each row shows an OPERA chain (left most) and the compact dependency graphs on each node (right). The fork events are shown in red and green vertices

\subsection{Root Selection}
All nodes can create event blocks and an event block can be a root when satisfying specific conditions. Not all event blocks can be roots. First, the first created event blocks are themselves roots. These leaf event blocks form the first root set $R_{S_1}$ of the first frame $f_1$. If there are total $n$ nodes and these nodes create the event blocks, then the cardinality of the first root set $|R_{S_1}|$ is $n$. Second, if an event block $e$ can reach at least 2n/3 roots, then $e$ is called a root. This event $e$ does not belong to $R_{S1}$, but the next root set $R_{S_2}$ of the next frame $f_2$. Thus, excluding the first root set, the range of cardinality of root set $R_{S_k}$ is $2n/3 < |R_{S_k}| \leq n$. The event blocks including $R_{S_k}$ before $R_{S_{k+1}}$ is in the frame $f_k$. The roots in $R_{S_{k+1}}$ does not belong to the frame $f_k$. Those are included in the frame $f_{k+1}$ when a root belonging to $R_{S_{k+2}}$ occurs.  

We introduce the use of a flag table to quickly determine whether a new event block becomes a root. Each node maintains a flag table of the top event block.
Every event block that is newly created is assigned $k$ hashes for its $k$ referenced event blocks. We apply an $OR$ operation on each set in the flag table of the referenced event blocks.

\begin{figure} [H] \centering  
\includegraphics[height=6cm, width=1.0\columnwidth]{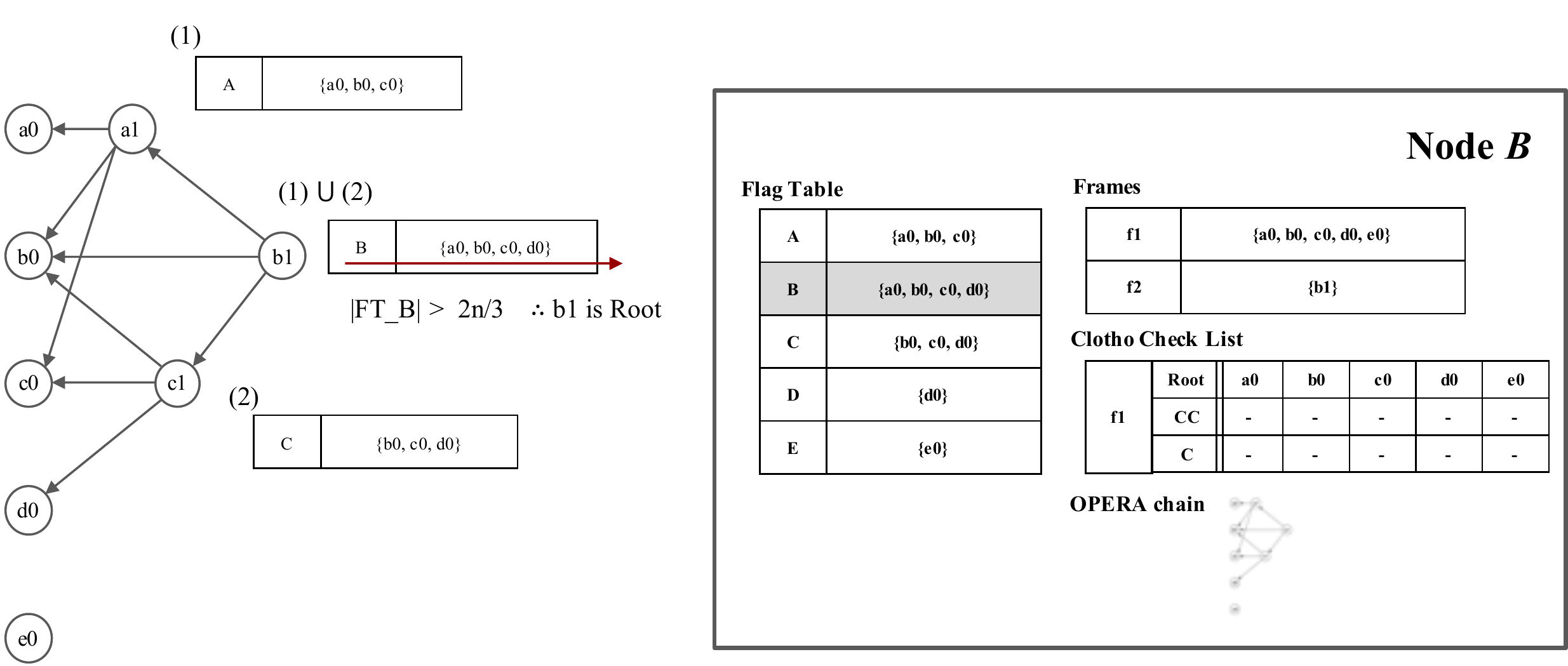}
\caption{An Example of Root selection}
\label{fig:ex_ft}
\end{figure}

Figure~\ref{fig:ex_ft} shows an example of how to use flag tables to determine a root. In this example, $b_1$ is the most recently created event block. We apply an $OR$ operation on each set of the flag tables for $b_1$'s $k$ referenced event blocks. The result is the flag table of $b_1$.  If the cardinality of the root set in $b_1$'s flag table is more than $2n/3$, $b_1$ is a root. In this example, the cardinality of the root set in $b_1$ is 4, which is greater than $2n/3$ ($n$=5). Thus, $b_1$ becomes root. In this example, $b_1$ is added to frame $f_2$ since $b_1$ becomes new root. 

The root selection algorithm is as follows:

\begin{enumerate}
\item The first event blocks are considered as roots. 
\item When a new event block is added in the OPERA chain (DAG), we check whether the event block is a root by applying an $OR$ operation on each set of the flag tables connected to the new event block. If the cardinality of the root set in the flag table for the new event block is more than 2n/3, the new event block becomes a root. 
\item When a new root appears on the OPERA chain, nodes update their frames. If one of the new event blocks becomes a root, all nodes that share the new event block add the hash value of the event block to their frames.  
\item The new root set is created if the cardinality of the previous root set $R_{S_p}$ is more than 2n/3 and the new event block can reach $2n/3$ roots in $R_{S_p}$.
\item When the new root set $R_{S_{k+1}}$ is created, the event blocks from the previous root set $R_{S_k}$ to before $R_{S_{k+1}}$ belong to the frame $f_k$.
\end{enumerate}

\subsection{Clotho Selection}

A Clotho is a root that satisfies the Clotho creation conditions. 
Clotho creation conditions are that more than 2n/3 nodes know the root and a root knows this information.

In order for a root $r$ in frame $f_i$ to become a Clotho, $r$ must be reached by more than n/3 roots in the frame $f_{i+1}$. Based on the definition of the root, each root reaches more than 2n/3 roots in previous frames. If more than n/3 roots in the frame $f_{i+1}$ can reach $r$, then $r$ is spread to all roots in the frame $f_{i+2}$. It means that all nodes know the existence of $r$. If we have any root in the frame $f_{i+3}$, a root knows that $r$ is spread to more than 2n/3 nodes. It satisfies Clotho creation conditions. 

\begin{figure} [H] \centering  
	\includegraphics[width=0.5\textwidth]{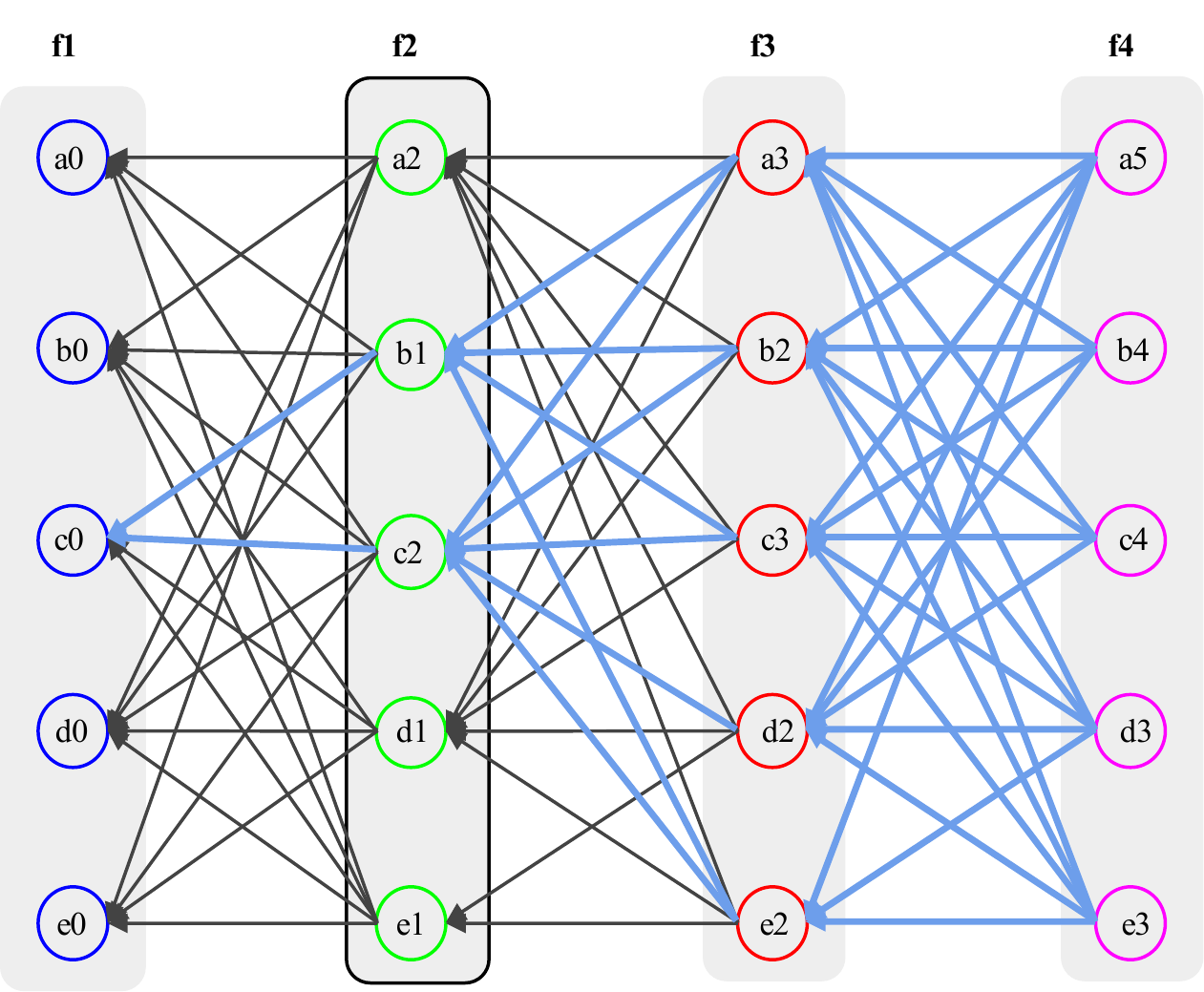}
	\caption{A verification of more than 2n/3 nodes}
	\label{fig:frame4}
\end{figure}

In the example in Figure~\ref{fig:frame4}, n is 5 and each circle indicates a root in a frame. Each arrow means one root can reach (happened-before) to the previous root. Each root has 4 or 5 arrows (out-degree) since n is 5 (more than 2n/3 $\geq$ 4). $b_1$ and $c_2$ in frame $f_2$ are roots that can reach $c_0$ in frame $f_1$. $d_1$ and $e_1$ also can reach $c_0$, but we only marked $b_1$ and $c_2$ (when n is 5, more than n/3 $\geq$ 2) since we show at least more than n/3 conditions in this example. And it was marked with a blue bold arrow (Namely, the roots that can reach root $c_0$ have the blue bold arrow). In this situation, an event block must be able to reach $b_1$ or $c_2$ in order to become a root in frame $f_3$ (In our example, n=5, more than n/3 $\geq$ 2, and more than 2n/3 $\geq$ 4. Thus, to be a root, either must be reached). All roots in frame $f_3$ reach $c_0$ in frame $f_1$.  

To be a root in frame $f_4$, an event block must reach more than 2n/3 roots in frame $f_3$ that can reach $c_0$. Therefore, if any of the root in frame $f_4$ exists, the root must have happened-before more than 2n/3 roots in frame $f_3$. Thus, the root of $f_4$ knows that $c_0$ is spread over more than 2n/3 of the entire nodes. Thus, we can select $c_0$ as Clotho.

\begin{figure}[H]\centering  
	\includegraphics[width=0.8\textwidth]{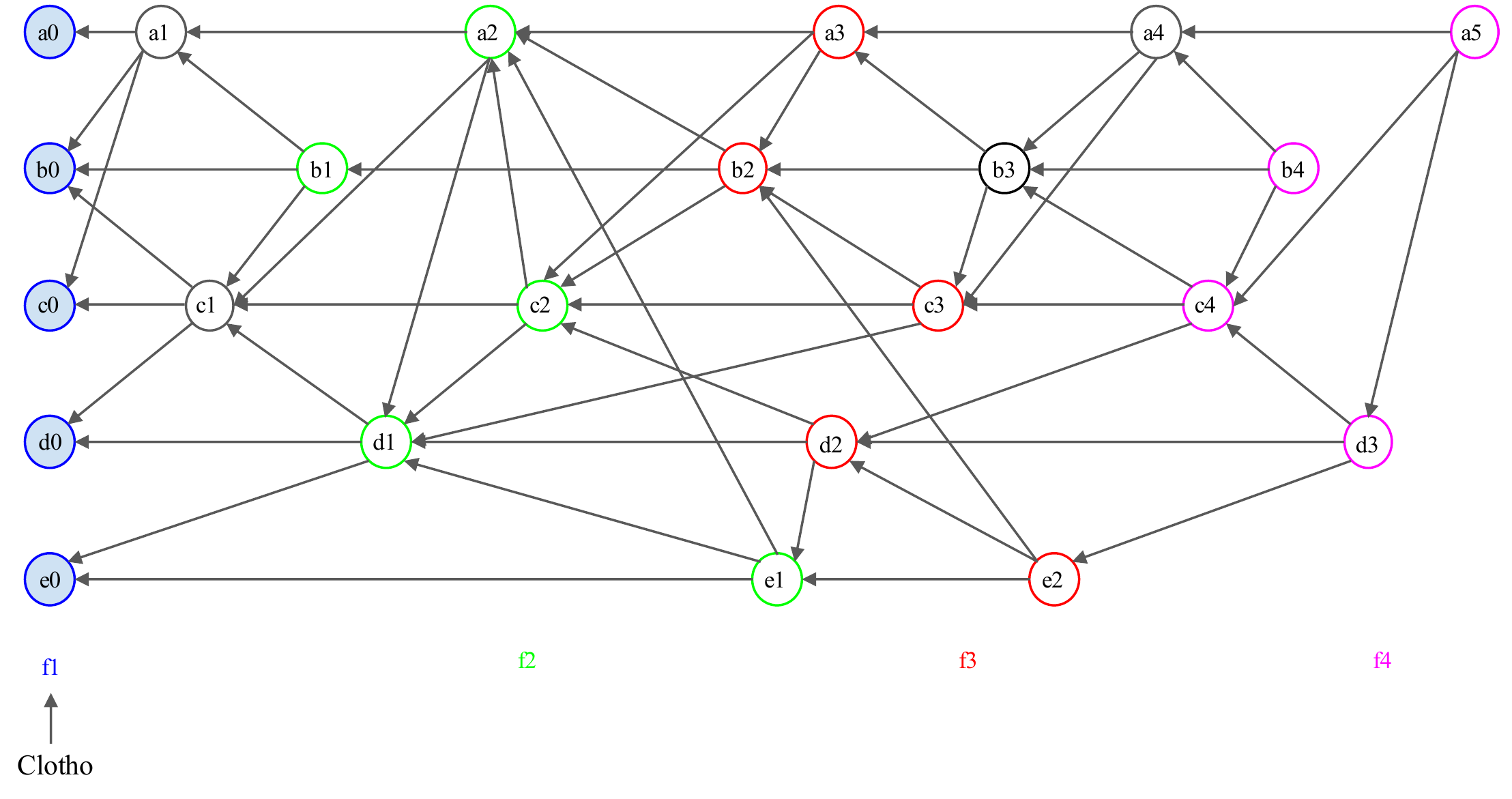}
	\caption{An Example of Clotho}
	\label{fig:Clotho}
\end{figure}

Figure~\ref{fig:Clotho} shows an example of a Clotho. In this example, all roots in the frame $f_1$ have happened-before more than n/3 roots in the frame $f_2$. We can select all roots in the frame $f_1$ as Clotho since the recent frame is $f_4$. 

\begin{algorithm} [H]
\caption{Clotho Selection}\label{al:acs}
\begin{algorithmic}[1]
	\Procedure{Clotho Selection}{}
	\State \textbf{Input}: a root $r$
	\For{$c$ $\in$ $frame(i-3, r)$} 
	\State$c.is\_clotho$ $\leftarrow$ $nil$ 
	\State$c.yes$ $\leftarrow$ 0
	\For{$c'$ $\in$ $frame(i-2, r)$} 
	\If{$c'$ has happened-before $c$} 
	\State c.yes $\leftarrow$ c.yes + 1
	\EndIf
	\EndFor
	\If{$c.yes > 2n/3$}
	\State $c.is\_clotho$ $\leftarrow$ $yes$
	\EndIf
	\EndFor
	\EndProcedure
\end{algorithmic}
\end{algorithm}

Algorithm~\ref{al:acs} shows the pseudo code for Clotho selection. The algorithm takes a root $r$ as input. 
Line 4 and 5 set $c.is\_clotho$ and $c.yes$ to $nil$ and 0 respectively. Line 6-8 checks whether any root $c'$ in $frame(i-3,r)$ has happened-before with the 2n/3 condition $c$ where $i$ is the current frame. In line 9-10, if the number of roots in $frame(i-2,r)$ which happened-before $c$ is more than $2n/3$, the root $c$ is set as a Clotho. The time complexity of Algorithm 3 is $O(n^{2})$, where $n$ is the number of nodes. 

\begin{figure}[H] \centering  
	\includegraphics[width=1.0\columnwidth]{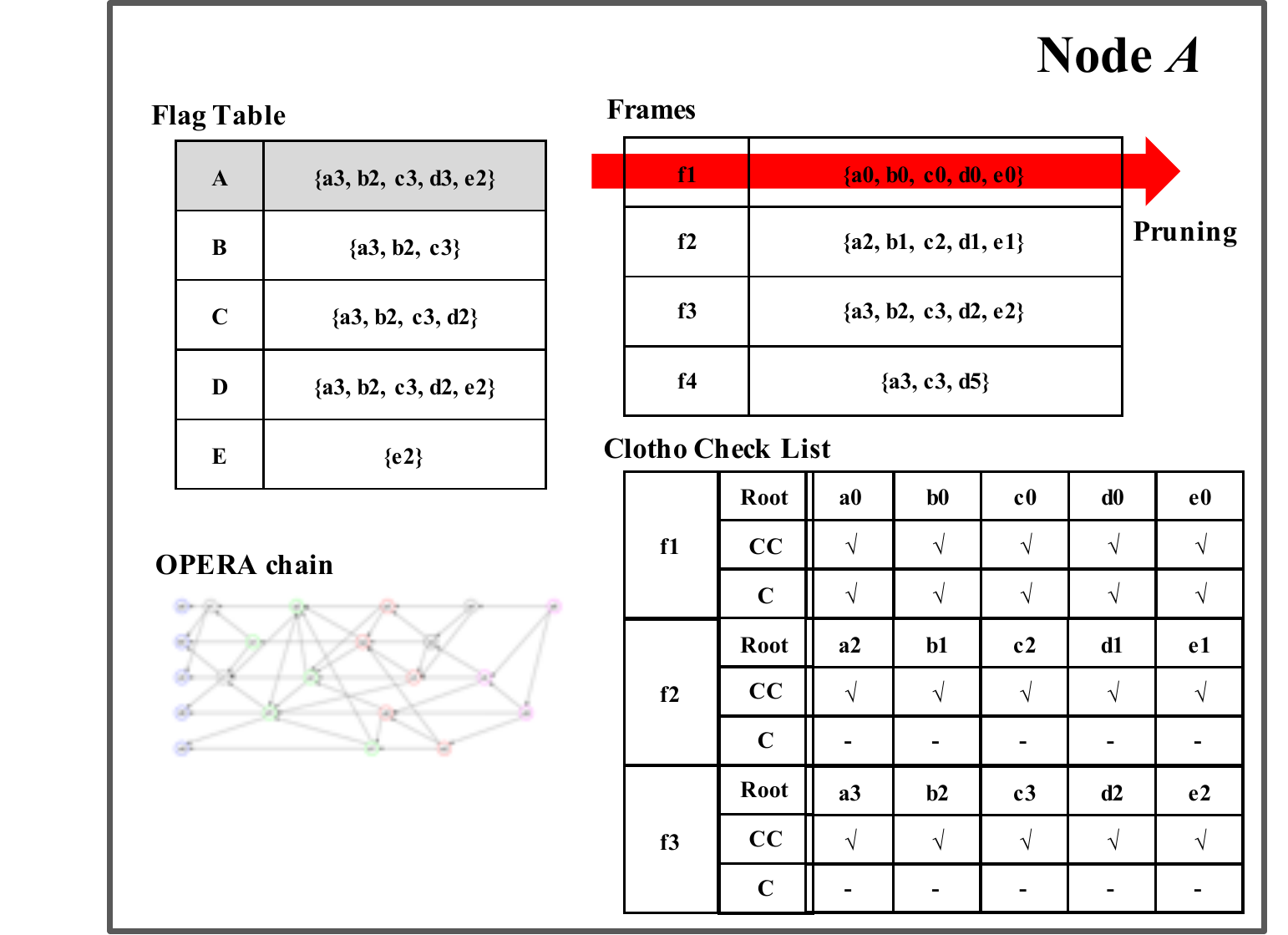}
	\caption{The node of \textit{A} when Clotho is selected}
	\label{fig:ClothoNodeA}	
\end{figure}

Figure~\ref{fig:ClothoNodeA} shows the state of node \textit{A} when a Clotho is selected. In this example, node \textit{A} knows all roots in the frame $f_1$ become Clotho's. Node \textit{A} prunes unnecessary information on its own structure. In this case, node \textit{A} prunes the root set in the frame $f_1$ since all roots in the frame $f_1$ become Clotho and the Clotho Check list stores the Clotho information. 

\subsection{Atropos Selection}

Atropos selection algorithm is the process in which the candidate time generated from Clotho selection is shared with other nodes, and each root re-selects candidate time repeatedly until all nodes have same candidate time for a Clotho. 

After a Clotho is nominated, each node then computes a candidate time of the Clotho. If there are more than two-thirds of the nodes that compute the same value for candidate time, that time value is recorded. Otherwise, each node reselects candidate time. By the reselection process, each node reaches time consensus for candidate time of Clotho as the OPERA chain (DAG) grows. The candidate time reaching the consensus is called Atropos consensus time. After Atropos consensus time is computed, the Clotho is nominated to Atropos and each node stores the hash value of Atropos and Atropos consensus time in Main-Chain (blockchain). The Main-chain is used for time order between event blocks. The proof of Atropos consensus time selection is shown in the section~\ref{se:proof}. 

\begin{figure}[H] \centering  
	\includegraphics[height=8cm, width=1.0\columnwidth]{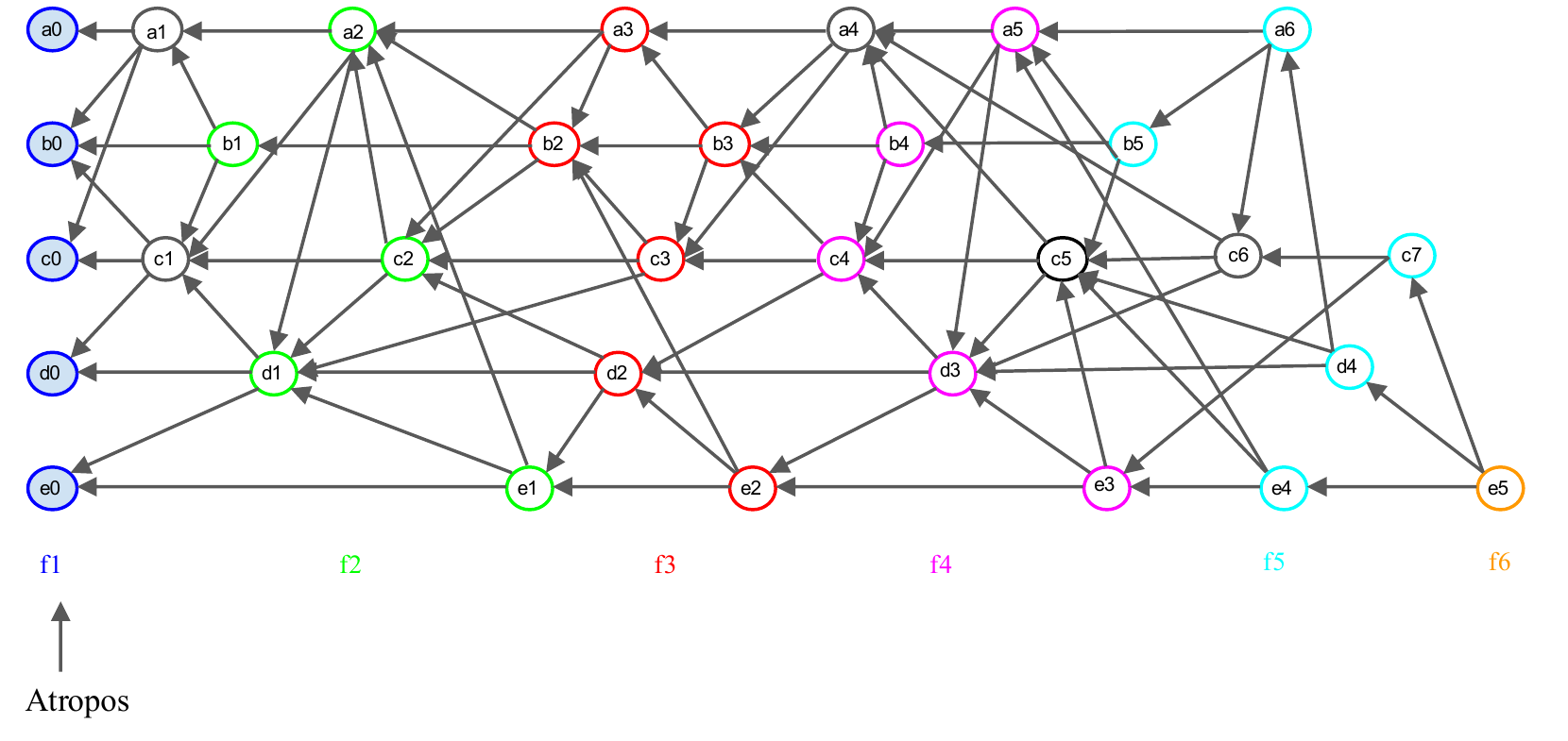}
	\caption{An Example of Atropos}
	\label{fig:Atropos}
\end{figure}

Figure~\ref{fig:Atropos} shows the example of Atropos selection. In  Figure~\ref{fig:Clotho}, all roots in the frame $f_1$ are selected as Clotho through the existence of roots in the frame $f_4$. Each root in the frame $f_5$ computes candidate time using timestamps of reachable roots in the frame $f_4$. Each root in the frame $f_5$ stores the candidate time to min-max value space. The root $r_6$ in the frame $f_6$ can reach more than 2n/3 roots in $f_5$ and $r_6$ can know the candidate time of the reachable roots that $f_5$ takes. If $r_6$ knows  the same candidate time than more than 2n/3, we select the candidate time as Atropos consensus time. Then all Clotho in the frame $f_1$ become Atropos. 

\newpage

\begin{figure}[H] \centering  
	\includegraphics[width=0.8\columnwidth]{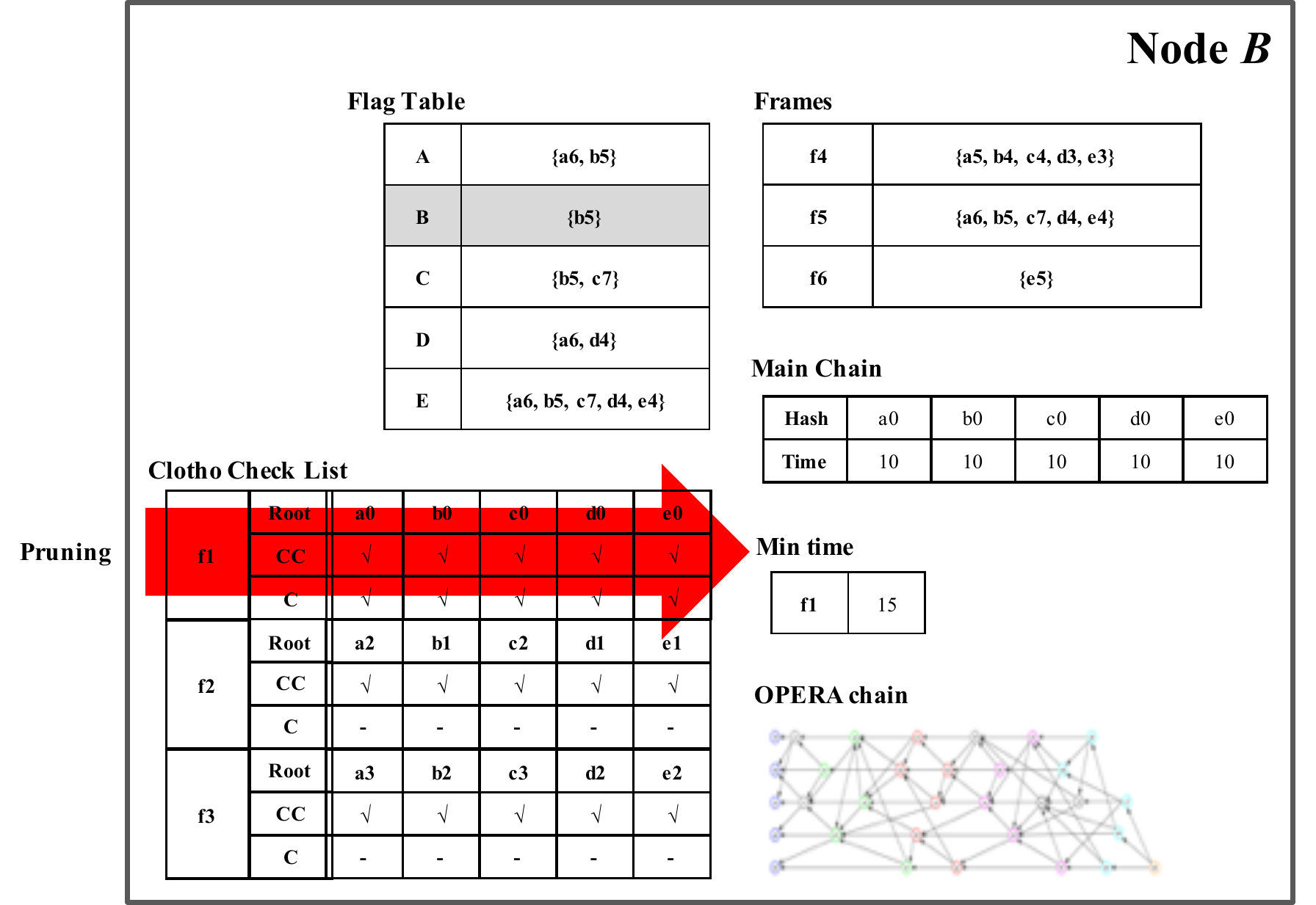}
	\caption{The node of \textit{B} when Atropos is selected}
	\label{fig:Atropos_Node}
\end{figure}

Figure~\ref{fig:Atropos_Node} shows the state of node \textit{B} when Atropos is selected. In this example, node \textit{B} knows all roots in the frame $f_1$ become Atropos. Then node \textit{B} prunes 
information of the frame $f_1$ in clotho check list since all roots in the frame $f_1$ become Atropos and main chain stores Atropos information. 

\begin{algorithm} [H]
\caption{Atropos Consensus Time Selection}\label{al:atc}
\begin{algorithmic}[1]
	\Procedure{Atropos Consensus Time Selection}{}
	\State \textbf{Input}: $c.Clotho$ in frame $f_i$
	\State$c.consensus\_time$ $\leftarrow$ $nil$
	\State$m$ $\leftarrow$ the index of the last frame $f_m$
	\For{d from 3 to (m-i)}
	\State $R$ $\leftarrow$ be the Root set $R_{S_{i+d}}$ in frame $f_{i+d}$
	\For{$r$ $\in$ $R$} 
	\If{d is 3}
	\If{$r$ confirms $c$ as Clotho}
	\State $r.time(c)$ $\leftarrow$ $r.lamport\_time$
	\EndIf
	\ElsIf{d $>$ 3}
	\State s $\leftarrow$ the set of Root in $f_{j-1}$ that $r$ can be  happened-before with 2n/3 condition
	\State t $\leftarrow$ RESELECTION(s, $c$)
	\State k $\leftarrow$ the number of root having $t$ in $s$
	\If{d mod $h$ $>$ 0}
	\If{$k$ $>$ 2n/3}
	\State $c.consensus\_time$ $\leftarrow$ $t$
	\State $r.time(c)$ $\leftarrow$ $t$
	\Else
	\State $r.time(c)$ $\leftarrow$ $t$
	\EndIf
	\Else
	\State $r.time(c)$ $\leftarrow$ the minimum value in $s$
	\EndIf
	\EndIf
	\EndFor
	\EndFor
	\EndProcedure
\end{algorithmic}
\end{algorithm}

\begin{algorithm} [H]
\caption{Consensus Time Reselection}\label{al:resel}
\begin{algorithmic}[1]
	\Function{Reselection}{} \funclabel{alg:a}
	\State \textbf{Input}: Root set $R$, and Clotho $c$
	\State \textbf{Output}: candidate time $t$
	\State $\tau$ $\leftarrow$ set of all $t_i = r.time(c)$ for all $r$ in $R$
	\State $D$  $\leftarrow$ set of tuples $(t_i, c_i)$ computed from $\tau$, where $c_i= count(t_i)$
	\State $max\_count$ $\leftarrow$ $max(c_i)$
	\State $t$ $\leftarrow$ $infinite$
	\For{tuple $(t_i, c_i)$ $\in$ $D$}
	\If{$max\_count$ $==$ $c_i$ $\&\&$ $t_i$ $<$ $t$}
	\State $t$ $\leftarrow$ $t_i$
	\EndIf
	\EndFor
	\State \textbf{return} $t$
	\EndFunction
\end{algorithmic}
\end{algorithm}

Algorithm~\ref{al:atc} and~\ref{al:resel} show pseudo code of Atropos consensus time selection and Consensus time reselection. In Algorithm~\ref{al:atc}, at line 6, $d$ saves the deference of relationship between root set of $c$ and $w$. Thus, line 8 means that $w$ is one of the elements in root set of the frame $f_{i+3}$, where the frame $f_i$ includes $c$. Line 10, each root in the frame $f_j$ selects own Lamport timestamp as candidate time of $c$ when they confirm root $c$ as Cltoho. In line 12, 13, and 14, $s$, $t$, and $k$ save the set of root that $w$ can be happened-before with 2n/3 condition $c$, the result of $RESELECTION$ function, and the number of root in $s$ having $t$. Line 15 is checking whether there is a difference as much as $h$ between $i$ and $j$ where $h$ is a constant value for minimum selection frame. Line 16-20 is checking whether more than two-thirds of root in the frame $f_{j-1}$ nominate the same candidate time. If two-thirds of root in the frame $f_{j-1}$ nominate the same candidate time, the root $c$ is assigned consensus time as $t$. Line 22 is minimum selection frame. In minimum selection frame, minimum value of candidate time is selected to reach byzantine agreement. Algorithm~\ref{al:resel} operates in the middle of Algorithm~\ref{al:atc}. In Algorithm~\ref{al:resel}, input is a root set $W$ and output is a reselected candidate time. Line 4-5 computes the frequencies of each candidate time from all the roots in $W$. In line 6-11, a candidate time which is smallest time that is the most nomitated. The time complexity of Algorithm~\ref{al:resel} is $O(n)$ where $n$ is the number of nodes. Since Algorithm~\ref{al:atc} includes Algorithm~\ref{al:resel}, the time complexity of Algorithm~\ref{al:atc} is $O(n^2)$ where $n$ is the number of nodes.

In the Atropos Consensus Time Selection algorithm, nodes reach consensus agreement about candidate time of a Clotho without additional communication (i.e., exchanging candidate time) with each other. Each node communicates with each other through the Lachesis protocol, the OPERA chain of all nodes grows up into same shape. This allows each node to know the candidate time of other nodes based on its OPERA chain and reach a consensus agreement. The proof that the agreement based on OPERA chain become agreement in action is shown in the section~\ref{se:proof}.

Atropos can be determined by the consensus time of each Clotho. It is an event block that is determined by finality and is non-modifiable. Furthermore, all event blocks can be reached from Atropos guarantee finality.

\newpage
\subsection{Lachesis Consensus}

\begin{figure}[H] \centering
\includegraphics[height=7cm, width=1.0\columnwidth]{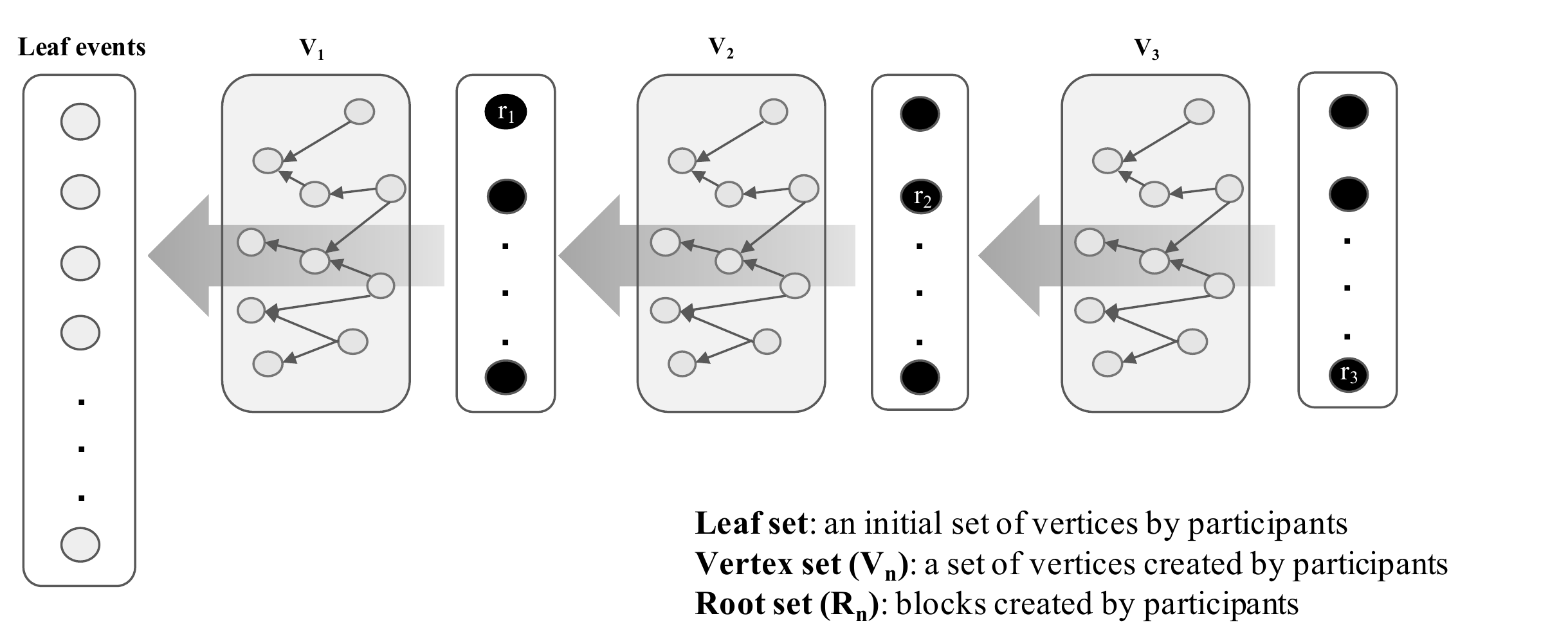}
\caption{Consensus Method in a DAG (combines chain with consensus process of pBFT)}
\label{fig:pBFTtoPath}
\end{figure}
Figure~\ref{fig:pBFTtoPath} illustrates how consensus is reached through the domination relation in the OPERA chain. In the figure, leaf set, denoted by $R_{s0}$, consists of the first event blocks created by individual participant nodes. $V$ is the set of event blocks that do not belong neither in $R_{s0}$ nor in any root set $R_{si}$.
Given a vertex $v$ in $V \cup R_{si}$, there exists a path from $v$ that can reach a leaf vertex $u$ in $R_{s0}$. 
Let $r_1$ and $r_2$ be root event blocks in root set $R_{s1}$ and $R_{s2}$, respectively.
$r_1$ is the block where a quorum or more blocks exist on a path that reaches a leaf event block. 
Every path from $r_1$ to a leaf vertex will contain a vertex in $V_1$. Thus, if there exists a vertex $r$ in $V_1$ such that $r$ is created by more than a quorum of participants, then $r$ is already included in $R_{s1}$. Likewise, $r_2$ is a block that can be reached for $R_{s1}$ including $r_1$ through blocks made by a quorum of participants.
For all leaf event blocks that could be reached by $r_1$, they are connected with more than quorum participants through the presence of $r_1$. The existence of the root $r_2$ shows that information of $r_1$ is connected with more than a quorum. 
This kind of a path search allows the chain to reach consensus in a similar manner as the pBFT consensus processes. It is essential to keep track of the blocks satisfying the pBFT consensus process for quicker path search; our OPERA chain and Main-chain keep track of these blocks.

The sequential order of each event block is an important aspect for Byzantine fault tolerance. In order to determine the pre-and-post sequence between all event blocks, we use Atropos consensus time, Lamport timestamp algorithm and the hash value of the event block.

\newpage

\begin{figure}[H] \centering  
	\includegraphics[width=1.0\textwidth]{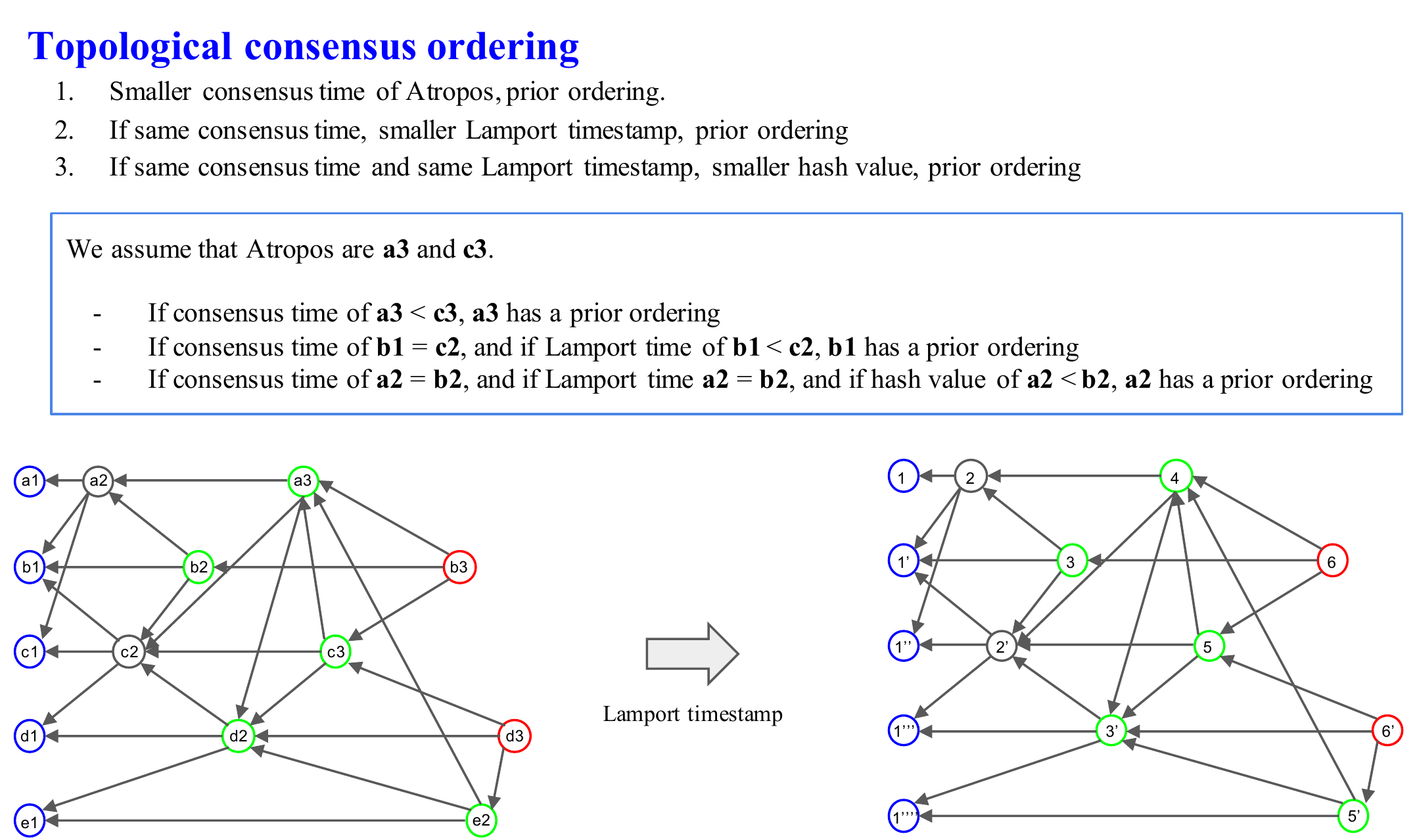}
	\caption{An example of topological consensus ordering}
	\label{fig:topological consensus ordering}
\end{figure}

First, when each node creates event blocks, they have a logical timestamp based on Lamport timestamp. This means that they have a partial ordering between the relevant event blocks. 
Each Clotho has consensus time to the Atropos. This consensus time is computed based on the logical time nominated from other nodes at the time of the 2n/3 agreement.

Each event block is based on the following three rules to reach an agreement:

\begin{enumerate}
\item If there are more than one Atropos with different times on the same frame, the event block with smaller consensus time has higher priority.
\item If there are more than one Atropos having any of the same consensus time on the same frame, determine the order based on the own logical time from Lamport timestamp.
\item When there are more than one Atropos having the same consensus time, if the local logical time is same, a smaller hash value is given priority through hash function.
\end{enumerate}

\newpage

\begin{figure}[H] \centering  
	\includegraphics[width=0.9\textwidth]{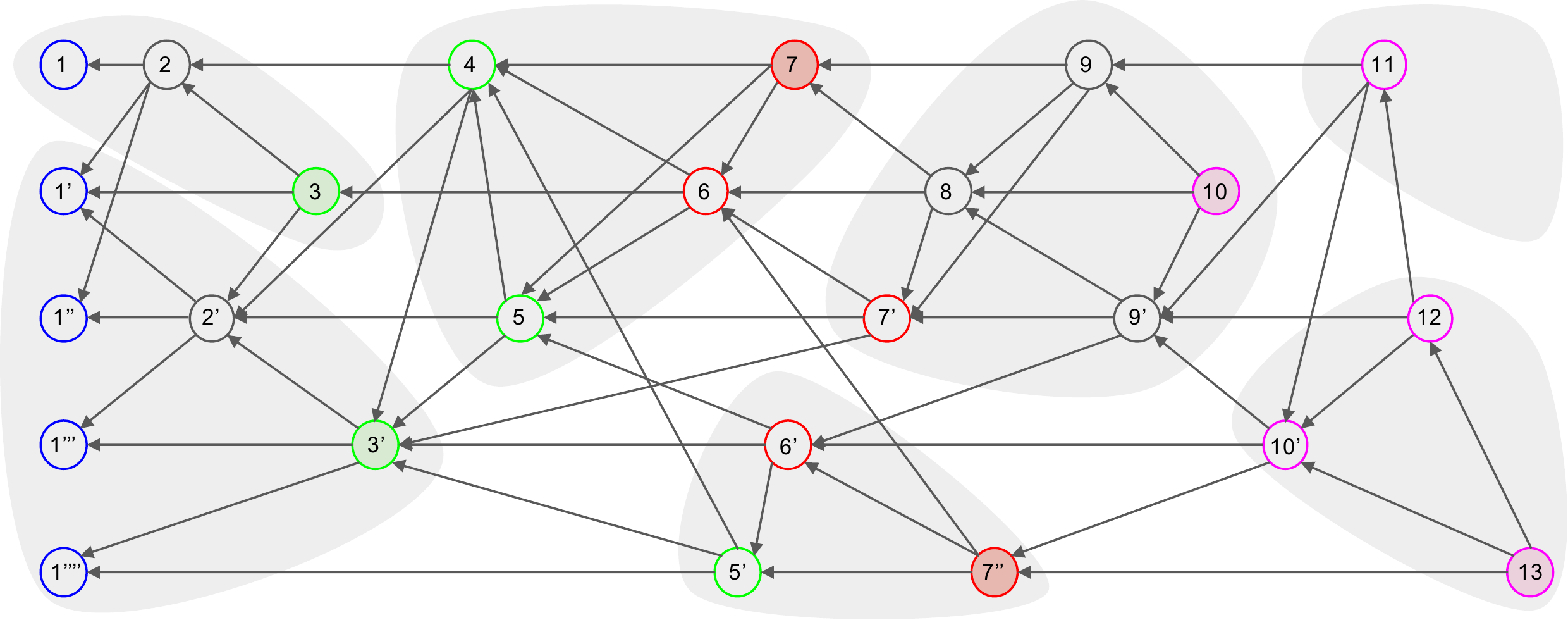}
	\caption{An Example of time ordering of event blocks in OPERA chain}
	\label{fig:sequence of operachain}
\end{figure}

Figure~\ref{fig:sequence of operachain} shows the part of OPERA chain in which the final consensus order is determined based on these 3 rules. The number represented by each event block is a logical time based on Lamport timestamp. Final topological consensus order containing the event blocks are based on agreement from the apropos. Based on each Atropos, they will have different colors depending on their range.

\subsection{Detecting Forks}

\dfnn{Fork}{A pair of events ($v_x$, $v_y$) is a fork if $v_x$ and $v_y$ have the same creator, but neither is a self-ancestor of the other. Denoted by $v_x \efork v_y$.}

For example, let $v_z$ be an event in node $n_1$ and two child events $v_x$ and $v_y$ of $v_z$. if $v_x \eself v_z$, $v_y \eself v_z$, $v_x \not \eself v_y$, $v_y \not \eself v_z$, then ($v_x$, $v_y$) is a fork.
The fork relation is symmetric; that is $v_x \efork v_y$ iff $v_y \efork v_x$.

By definition, ($v_x$, $v_y$) is a fork if $cr(v_x)=cr(v_y)$, $v_x \not \eancestor v_y$ and $v_y \not \eancestor v_x$. Using Happened-Before, the second part means $v_x \not \rightarrow v_y$ and $v_y \not \rightarrow v_x$. By definition of concurrent, we get $v_x \concur v_y$.

\begin{lem}
If there is a fork $v_x \efork  v_y$, then $v_x$ and $v_y$ cannot both be roots on honest nodes.
\end{lem}
Here, we show a proof by contradiction. Any honest node cannot accept a fork so $v_x$ and $v_y$ cannot be roots on the same honest node. Now we prove a more general case. Suppose that both $v_x$ is a root of $n_x$ and $v_y$ is root of $n_y$, where $n_x$ and $n_y$ are honest nodes. Since $v_x$ is a root, it reached events created by more than 2/3 of member nodes. Similarly, $v_y$ is a root, it reached events created by more than 2/3 of member nodes. Thus, there must be an overlap of more than $n$/3 members of those events in both sets. Since we assume less than $n$/3 members are not honest, so there must be at least one honest member in the overlap set. Let $n_m$ be such an honest member. Because $n_m$ is honest, $n_m$ does not allow the fork.

\section{Conclusion}\label{se:con}
We further optimize the OPERA chain and Main-chain for faster consensus. By using Lamport timestamps and domination relation, the topological ordering of event blocks in OPERA chain and Main chain is more intuitive and reliable in distributed system. 

We have presented a formal semantics for Lachesis protocol in Section~\ref{se:lca}.
Our formal proof of pBFT for our Lachesis protocol is given in Section~\ref{se:proof}.  Our work is the first that studies such concurrent common knowledge sematics~\cite{cck92} and dominator relationships in DAG-based protocols.

\newpage
\section{Appendix}\label{se:appendix}

\subsection{Preliminaries}
The history of a Lachesis protocol can be represented by a directed acyclic graph $G=(V, E)$, where $V$ is a set of vertices and $E$ is a set of edges. Each vertex in a row (node) represents an event. Time flows left-to-right of the graph, so left vertices represent earlier events in history.
A path $p$ in $G$ is a sequence  of vertices ($v_1$, $v_2$, $\dots$, $v_k$) by following the edges in $E$.
Let $v_c$ be a vertex in $G$.
A vertex $v_p$ is the \emph{parent} of $v_c$ if there is an edge from $v_p$ to $v_c$.
A vertex $v_a$ is an \emph{ancestor} of $v_c$ if there is a path from $v_a$ to $v_c$.

\begin{defn}[node]
	Each machine that participates in the Lachesis protocol is called a node. \end{defn}

Let $n$ denote the total number of nodes.

\begin{defn}[event block]
	Each node can create event blocks, send (receive) messages to (from) other nodes.
\end{defn}

\begin{defn}[vertex]
	An event block is a vertex of the OPERA chain.
\end{defn}

Suppose a node $n_i$ creates an event $v_c$ after an event $v_s$ in $n_i$.  Each event block has exactly $k$ references. One of the references is self-reference, and the other $k$-1 references point to the top events of $n_i$'s $k$-1 peer nodes.

\begin{defn}[peer node]
	A node $n_i$ has $k$ peer nodes.
\end{defn}

\begin{defn}[top event]
	An event $v$ is a top event of a node $n_i$ if there is no other event in $n_i$ referencing $v$.
\end{defn}

\begin{defn}[self-ref]
	An event $v_s$ is called ``self-ref" of event $v_c$, if the self-ref hash of $v_c$ points to the event $v_s$. Denoted by $v_c \eself v_s$.
\end{defn}

\begin{defn}[ref]
	An event $v_r$ is called ``ref" of event $v_c$ if the reference hash of $v_c$ points to the event $v_r$. Denoted by $v_c \eref v_r$.
\end{defn}

For simplicity, we can use $\erefz$ to denote a reference relationship (either $\eref$ or $\eself$).

\begin{defn}[self-ancestor]
	An event block $v_a$ is self-ancestor of an event block $v_c$ if there is a sequence of events such that $v_c \eself v_1 \eself \dots \eself v_m \eself v_a $. Denoted by $v_c \eselfancestor v_a$.
\end{defn}

\begin{defn}[ancestor]
	An event block $v_a$ is an ancestor of an event block $v_c$ if there is a sequence of events such that $v_c \erefz v_1 \erefz \dots \erefz v_m \erefz v_a $. Denoted by $v_c \eancestor v_a$.
\end{defn}

For simplicity, we simply use $v_c \eancestor v_s$ to refer both ancestor and self-ancestor relationship, unless we need to distinguish the two cases.

\begin{defn}[OPERA chain]
	OPERA chain is a DAG graph $G = (V, E)$ consisting of $V$ vertices and $E$ edges. Each vertex $v_i \in V$ is an event block. An edge $(v_i,v_j) \in E$ refers to a hashing reference from $v_i$ to $v_j$; that is, $v_i \erefz v_j$.
\end{defn}

\subsubsection{Domination relation}
Then we define the domination relation for event blocks. To begin with, we first introduce pseudo vertices, \emph{top} and \emph{bot}, of the DAG OPERA chain $G$.
\begin{defn}[pseudo top]
	A pseudo vertex, called top, is the parent of all top event blocks. Denoted by $\top$.
\end{defn}
\begin{defn}[pseudo bottom]
	A pseudo vertex, called bottom, is the child of all leaf event blocks. Denoted by $\bot$.
\end{defn}

With the pseudo vertices, we have $\bot$ happened before all event blocks. Also all event blocks happened before $\top$. That is, for all event $v_i$, $\bot \hbefore v_i$ and $v_i \hbefore \top$.

\begin{defn}[dom]
	An event $v_d$ dominates an event $v_x$ if every path from $\top$ to $v_x$ must go through $v_d$. Denoted by $v_d \dom v_x$.
\end{defn}

\begin{defn}[strict dom]
	An event $v_d$ strictly dominates an event $v_x$ if $v_d \dom v_x$ and $v_d$ does not equal $v_x$. Denoted by $v_d \sdom v_x$.
\end{defn}

\begin{defn}[domfront]
	A vertex $v_d$ is said ``domfront'' a vertex $v_x$ if  $v_d$ dominates an immediate predecessor of $v_x$, but $v_d$ does not strictly dominate $v_x$. Denoted by $v_d \domf v_x$.
\end{defn}

\begin{defn}[dominance frontier]
	The dominance frontier of a vertex $v_d$ is the set of all nodes $v_x$ such that $v_d \domf v_x$. Denoted by $DF(v_d)$.
\end{defn}

From the above definitions of domfront and dominance frontier, the following holds. If $v_d \domf v_x$, then $v_x \in DF(v_d)$.

Here, we introduce a new idea that extends the concept domination.

\begin{defn}[subgraph] 
	For a vertex $v$ in a DAG $G$, let $G[v] = (V_v,E_v)$ denote an induced-subgraph of $G$ such that $V_v$ consists of all ancestors of $v$ including $v$, and $E_v$ is the induced edges of $V_v$ in $G$.
\end{defn}

For a set $S$ of vertices, an event $v_d$  $\frac{2}{3}$-dominates $S$ if there are more than 2/3 of vertices $v_x$ in $S$ such that $v_d$ dominates $v_x$. 	
Recall that $R_1$ is the set of all leaf vertices in $G$. The $\frac{2}{3}$-dom set $D_0$ is the same as the set $R_1$.The $\frac{2}{3}$-dom set $D_i$ is defined as follows:	
\begin{defn}[$\frac{2}{3}$-dom set]	
	A vertex $v_d$ belongs to a $\frac{2}{3}$-dom set  within the graph $G[v_d]$, if $v_d$ $\frac{2}{3}$-dominates $R_1$.
	The $\frac{2}{3}$-dom set $D_k$ consists of all roots $d_i$ such that  $d_i$ $\not \in $ $D_i$, $\forall$ $i$ = 1..($k$-1), and $d_i$ $\frac{2}{3}$-dominates $D_{i-1}$.
\end{defn}

\begin{lem}
	The $\frac{2}{3}$-dom set $D_i$ is the same with the root set $R_i$, for all nodes.
\end{lem}

\subsection{Proof of Lachesis Consensus Algorithm}\label{se:proof}
This section presents a proof of liveness and safety  of our Lachesis protocols. We aim to show that our consensus is Byzantine fault tolerant with a presumption that more than two-thirds of participants are reliable nodes. We first provide some definitions, lemmas and theorems. Then we validate the Byzantine fault tolerance. 

\subsubsection{Proof of Byzantine Fault Tolerance for Lachesis Consensus Algorithm}

\begin{defn}[Happened-Immediate-Before]
An event block $v_x$ is said Happened-Immediate-Before an event block $v_y$ if $v_x$ is a (self-) ref of $v_y$. Denoted by $v_x \hibefore v_y$.
\end{defn}

\begin{defn}[Happened-Before]	
	 An event block $v_x$ is said Happened-Before an event block $v_y$ if $v_x$ is a (self-) ancestor of $v_y$. Denoted by $v_x \hbefore v_y$.
\end{defn}

The happens-before relation is the transitive closure of happens-immediately-before.
Thus, an event $v_x$ happened before an event $v_y$ if one of the followings happens: (a) $v_y \eself v_x$, (b) $v_y \eref v_x$,  or (c) $v_y \eancestor v_x$.
 We come up with the following proposition:
\begin{prop}[Happened-Immediate-Before OPERA]
$v_x \hibefore v_y$ iff $v_y \erefz v_x$ iff edge $(v_y, v_x)$ $\in E$ of OPERA chain.
\end{prop}
\begin{lem}[Happened-Before Lemma]
	$v_x \hbefore v_y$ iff $v_y \eancestor v_x$.
\end{lem}


\begin{defn}[concurrent]
	Two event blocks $v_x$ and $v_y$ are said concurrent if neither of them  happened before the other. Denoted by $v_x \concur v_y$.
\end{defn}

Given two vertices $v_x$ and $v_y$ both contained in two OPERA chains $G_1$ and $G_2$ on two nodes. We have the following: 
	(1) $v_x \hbefore v_y$ in $G_1$ iff $v_x \hbefore v_y$ in $G_2$; (2)
	$v_x \concur v_y$ in $G_1$ iff $v_x \concur v_y$ in $G_2$.

Below is some main definitions in Lachesis protocol.
\begin{defn}[Leaf]
	The first created event block of a node is called a leaf event block.
\end{defn}

\begin{defn}[Root]
\label{def:root}
	The leaf event block of a node is a root.
	When an event block $v$ can reach more than $2n/3$ of the roots in the previous frames, $v$ becomes a root.
\end{defn}

\begin{defn}[Root set]
	The set of all first event blocks (leaf events) of all nodes form the first root set $R_1$ ($|R_1|$ = $n$). The root set $R_k$ consists of all roots $r_i$ such that $r_i$ $\not \in $ $R_i$, $\forall$ $i$ = 1..($k$-1) and $r_i$ can reach more than 2n/3 other roots in the current frame, $i$ = 1..($k$-1).  
\end{defn}

\begin{defn}[Frame]
Frame $f_i$ is a natural number that separates Root sets. 
\end{defn} 

The root set at frame $f_i$ is denoted by $R_i$.

\begin{defn}[consistent chains]\label{dfn:conchains} OPERA chains $G_1$ and $G_2$ are consistent iff for any event $v$ contained in both chains, $G_1[v] = G_2[v]$. Denoted by $G_1 \sim G_2$.
\end{defn}
When two consistent chains contain the same event $v$, both chains contain the same set of ancestors for $v$, with the same reference and self-ref edges between those ancestors:
\begin{thm}\label{thm:conchains}
	All nodes have consistent OPERA chains.
\end{thm}
\begin{proof}
	 If two nodes have OPERA chains containing event $v$, then they have the same $k$ hashes contained within $v$. A node will not accept an event during a sync unless that node already has $k$ references for that event, so both OPERA chains must contain $k$ references for $v$. The cryptographic hashes are assumed to be secure, therefore the references must be the same. By induction, all ancestors of $v$ must be the same. Therefore, the two OPERA chains are consistent.
\end{proof}

\begin{defn}[creator] If a node $n_x$ creates an event block $v$, then the creator of $v$, denoted by $cr(v)$, is $n_x$.
\end{defn}

\begin{defn}[fork]
	The pair of events ($v_x$, $v_y$) is a fork if $v_x$ and $v_y$ have the same creator, but neither is a self-ancestor of the other. Denoted by $v_x \efork v_y$.
\end{defn}
For example, let $v_z$ be an event in node $n_1$ and two child events $v_x$ and $v_y$ of $v_z$. if $v_x \eself v_z$, $v_y \eself v_z$, $v_x \not \eself v_y$, $v_y \not \eself v_z$, then ($v_x$, $v_y$) is a fork.
The fork relation is symmetric; that is $v_x \efork v_y$ iff $v_y \efork v_x$.
\begin{lem}
	$v_x \efork v_y$ iff $cr(v_x)=cr(v_y)$ and $v_x \concur v_y$.
\end{lem}
\begin{proof}
By definition, ($v_x$, $v_y$) is a fork if $cr(v_x)=cr(v_y)$, $v_x \not \eancestor v_y$ and $v_y \not \eancestor v_x$. Using Happened-Before, the second part means $v_x \not \rightarrow v_y$ and $v_y \not \rightarrow v_x$. By definition of concurrent, we get $v_x \concur v_y$.
\end{proof}

\begin{lem} (fork detection). If there is a fork $v_x \efork  v_y$, then $v_x$ and $v_y$ cannot both be roots on honest nodes.
\end{lem}
\begin{proof}
	Here, we show a proof by contradiction. Any honest node cannot accept a fork so $v_x$ and $v_y$ cannot be roots on the same honest node. Now we prove a more general case. Suppose that both $v_x$ is a root of $n_x$ and $v_y$ is root of $n_y$, where $n_x$ and $n_y$ are honest nodes. Since $v_x$ is a root, it reached events created by more than 2/3 of member nodes. Similary, $v_y$ is a root, it reached events created by  more than 2/3 of member nodes. Thus, there must be an overlap of more than $n$/3 members of those events in both sets. Since we assume less than $n$/3 members are not honest, so there must be at least one honest member in the overlap set. Let $n_m$ be such an honest member. Because $n_m$ is honest, $n_m$ does not allow the fork. This contradicts the assumption. Thus, the lemma is proved.
\end{proof}

Each node $n_i$ has an OPERA chain $G_i$. We define a consistent chain from a sequence of OPERA chain $G_i$.
\begin{defn}[consistent chain] 
	A global consistent chain $G^C$ is a chain if $G^C \sim G_i$ for all $G_i$.
\end{defn}

We denote $G \sqsubseteq G'$ to stand for $G$ is a subgraph of $G'$.
\begin{lem}
	$\forall G_i$ ($G^C \sqsubseteq G_i$).
\end{lem}
\begin{lem}
	$\forall v \in G^C$ $\forall G_i$ ($G^C[v] \sqsubseteq G_i[v]$).
\end{lem}
\begin{lem}
	($\forall v_c \in G^C$) ($\forall v_p \in G_i$) (($v_p \hbefore v_c) \Rightarrow v_p \in G^C$).
\end{lem}

Now we state the following important propositions.
\begin{defn}[consistent root]
	Two chains $G_1$ and $G_2$ are root consistent, if for every $v$ contained in both chains, and $v$ is a root of $j$-th frame in $G_1$, then $v$ is a root of $j$-th frame in $G_2$.
\end{defn}
\begin{prop}
	If $G_1 \sim G_2$, then $G_1$ and $G_2$ are root consistent.
\end{prop}
\begin{proof}
	By consistent chains, if $G_1 \sim G_2$ and $v$ belongs to both chains, then $G_1[v]$ = $G_2[v]$.
	We can prove the proposition by induction. For $j$ = 0, the first root set is the same in both $G_1$ and $G_2$. Hence, it holds for $j$ = 0. Suppose that the proposition holds for every $j$ from 0 to $k$. We prove that it also holds for $j$= $k$ + 1.
	 Suppose that $v$ is a root of frame $f_{k+1}$ in $G_1$. 
	Then there exists a set $S$ reaching 2/3 of members in $G_1$ of frame $f_k$ such that $\forall u \in S$ ($u\hbefore v$). As $G_1 \sim G_2$, and $v$ in $G_2$, then $\forall u \in S$ ($u \in G_2$). Since the proposition holds for $j$=$k$, 
	As $u$ is a root of frame $f_{k}$ in $G_1$, $u$ is a root of frame $f_k$ in $G_2$. Hence, the set $S$ of 2/3 members $u$ happens before $v$ in $G_2$. So $v$ belongs to $f_{k+1}$ in $G_2$. The proposition is proved.
\end{proof}

From the above proposition, one can deduce the following:
\begin{lem}
		$G^C$ is root consistent with $G_i$ for all nodes.
 \end{lem}
Thus, all nodes have the same consistent root sets, which are the root sets in $G^C$. Frame numbers are consistent for all nodes.

\begin{lem}[consistent flag table] For any top event $v$ in both OPERA chains $G_1$ and $G_2$, and $G_1 \sim G_2$, then the flag tables of $v$ are consistent iff they are the same in both chains.
\end{lem}
\begin{proof}
	From the above lemmas, the root sets of $G_1$ and $G_2$ are consistent. If $v$ contained in $G_1$, and $v$ is a root of $j$-th frame in $G^1$, then $v$ is a root of $j$-th frame in $G_i$. Since $G_1 \sim G_2$, $G_1[v] = G_2[v]$. The reference event blocks of $v$ are the same in both chains. Thus the flag tables of $v$ of both chains are the same.
\end{proof}
Thus, all nodes have consistent flag tables.

\begin{defn}[Clotho]
	A root $r_k$ in the frame $f_{a+3}$ can nominate a root $r_a$ as Clotho if more than 2n/3 roots in the frame $f_{a+1}$ Happened-Before $r_a$ and $r_k$ Happened-Before the roots in the frame $f_{a+1}$.
\end{defn} 

\begin{lem}
	\label{lem:root}
	For any root set $R$, all nodes nominate same root into Clotho.
\end{lem}

\begin{proof}
	Based on Theorem~\ref{thm:same}, each node nominates a root into Clotho via the flag table. If all nodes have an OPERA chain with same shape, the values in flag table should be equal to each other in OPERA chain. Thus, all nodes nominate the same root into Clotho since the OPERA chain of all nodes has same shape.
\end{proof}

\begin{lem}
	\label{lem:resel}
	In the Reselection algorithm, for any Clotho, a root in OPERA chain selects the same consensus time candidate.
\end{lem}

\begin{proof}
	Based on Theorem~\ref{thm:same}, if all nodes have an OPERA chain with the same partial shape, a root in OPERA chain selects the same consensus time candidate by the Reselection algorithm.
\end{proof}


\begin{lem}[Fork Lemma]
	\label{lem:fork}
	If the pair of event blocks (x,y) is fork and a root has Happened-before the fork, this fork is detected in the Clotho selection process.
\end{lem}

\begin{proof}
We show a proof by contradiction. Assume that no node can detect the fork in the Clotho selection process.

Assume that there is a root $r_i$ that becomes Clotho in $f_i$, which was selected as Clotho by n/3 of the roots in $f_{i+1}$. More than 2n/3 roots in $f_{i+1}$ should have happened-before by a root in $f_{i+2}$. 
If a pair ($v_x$, $v_y$) is fork, There are two cases: (1) assume that $r_i$ is one of $v_x$ and $v_y$, (2) assume that $r_i$ can reach both $v_x$ and $v_y$.

Our proof for both two cases is as follows. 

Let $k$ denote that the number of roots in $f_{i+1}$ that can reach $r_i$ in $f_{i}$ ($\therefore$ n/3 $<$ $k$).

In order to select root in $f_{i+2}$, the root in $f_{i+2}$ should reach more than 2n/3 of roots in $f_{i+1}$ by Definition~\ref{def:root}. At the moment, assume that $l$ is the number of roots in $f_{i+1}$ that can be reached by the root in $f_{i+2}$ ($\therefore$ 2n/3 $<$ $l$).

At this time, n $<$ k + l ($\because$ n/3 + 2n/3 $<$ $k$ + $l$), there are (n - $k$ + $l$) roots in frame $f_{i+1}$ that should reach $r_i$ in $f_i$ and all roots in $f_{i+2}$ should reach at least n – $k$ + $l$ of roots in $f_{i+1}$. It means that all roots in $f_{i+2}$ know the existence of $r_i$. Therefore, the node that generated all the roots of $f_{i+2}$ detect the fork in the Clotho selection of $r_i$, which contradicts the assumption.

It can be covered two cases. If $r_i$ is part of the fork, we can detect in $f_{i+2}$. If there is fork ($v_x$, $v_y$) that can be reached by $r_i$, it also can be detected in $f_{i+2}$ since we can detect the fork in the Clotho selection of $r_i$ and it indicates that all event blocks that can be reached by $r_i$ are detected by the roots in $f_{i+2}$.
\end{proof}

\begin{lem}
	For a root $v$ happened-before a fork in OPERA chain, $v$ must see the fork before becoming Clotho.
\end{lem}

\begin{proof}
	Suppose that a node creates two event blocks ($v_x, v_y$), which forms a fork. To create two Clothos that can reach both events, the event blocks should reach by more than 2n/3 nodes. Therefore, the OPERA chain can structurally detect the fork before roots become Clotho.
\end{proof}

\begin{thm}[Fork Absence]
\label{thm:same}
All nodes grows up into same consistent OPERA chain $G^C$, which contains no fork.
\end{thm}
\begin{proof}
 Suppose that there are two event blocks $v_x$ and $v_y$ contained in both $G_1$ and $G_2$, and their path between $v_x$ and $v_y$ in $G_1$ is not equal to that in $G_2$. We can consider that the path difference between the nodes is a kind of fork attack. Based on Lemma~\ref{lem:fork}, if an attacker forks an event block, each chain of $G_i$ and $G_2$ can detect and remove the fork before the Clotho is generated. Thus, any two nodes have consistent OPERA chain. 
\end{proof}

\begin{defn}[Atropos]
	If the consensus time of Clotho is validated, the Clotho become an Atropos. 
\end{defn}

%

\begin{thm}
\label{thm:ct}
Lachesis consensus algorithm guarantees to reach agreement for the consensus time.
\end{thm}

\begin{proof}
For any root set $R$ in the frame $f_{i}$, time consensus algorithm checks whether more than 2n/3 roots in the frame $f_{i-1}$ selects the same value. However, each node selects one of the values collected from the root set in the previous frame by the time consensus algorithm and Reselection process. Based on the Reselection process, the time consensus algorithm can reach agreement. However, there is a possibility that consensus time candidate does not reach agreement~\cite{Fischer85}. To solve this problem, time consensus algorithm includes minimal selection frame per next $h$ frame. In minimal value selection algorithm, each root selects minimum value among values collected from previous root set. Thus, the consensus time reaches consensus by time consensus algorithm.
\end{proof}

\begin{thm}
\label{thm:bft}
If the number of reliable nodes is more than $2n/3$, event blocks created by reliable nodes must be assigned to consensus order.
\end{thm}

\begin{proof}
In OPERA chain, since reliable nodes try to create event blocks by communicating with every other nodes continuously, reliable nodes will share the event block $x$ with each other. If a root $y$ in the frame $f_{i}$ Happened-Before event block $x$ and more than n/3 roots in the frame $f_{i+1}$ Happened-Before the root $y$, the root $y$ will be nominated as Clotho and Atropos. Thus, event block $x$ and root $y$ will be assigned consensus time $t$. 

For an event block, assigning consensus time means that the validated event block is shared by more than 2n/3 nodes. Therefore, malicious node cannot try to attack after the event blocks are assigned consensus time. When the event block $x$ has consensus time $t$, it cannot occur to discover new event blocks with earlier consensus time than $t$.
There are two conditions to be assigned consensus time earlier than $t$ for new event blocks. First, a root $r$ in the frame $f_{i}$ should be able to share new event blocks. Second, the more than 2n/3 roots in the frame $f_{i+1}$ should be able to share $r$. Even if the first condition is satisfied by malicious nodes (e.g., parasite chain),
the second condition cannot be satisfied since at least 2n/3 roots in the frame $f_{i+1}$ are already created and cannot be changed. Therefore, after an event block is validated, new event blocks should not be participate earlier consensus time to OPERA chain. 
\end{proof}

\subsection{Semantics of Lachesis protocol}\label{sec:semantics} 

This section gives the formal semantics of Lachesis consensus protocol.
We use CCK model \cite{cck92} of an asynchronous system as the base of the semantics of our Lachesis protocol. Events are ordered based on Lamport's happens-before relation. 
In particular, we use Lamport’s theory to describe global states of an asynchronous system.

We present notations and concepts, which are important for Lachesis protocol. In several places, we adapt the notations and concepts of CCK paper to suit our Lachesis protocol. 

An asynchronous system consists of the following:
\begin{defn}[process]
	A process $p_i$ represents a machine or a node. The process identifier of $p_i$ is $i$. A set $P$ = \{1,...,$n$\} denotes the set of process identifiers.
\end{defn}
\begin{defn}[channel]
	A process $i$ can send messages to process $j$ if there is a channel ($i$,$j$). Let $C$ $\subseteq$ \{($i$,$j$) s.t. $i,j \in P$\} denote the set of channels.
\end{defn}
\begin{defn}[state]
	A local state of a process $i$ is denoted by $s_j^i$.
\end{defn}
A local state consists of a sequence of event blocks $s_j^i = v_0^i, v_1^i, \dots, v_j^i$. 

In a DAG-based protocol, each $v_j^i$ event block is valid only the reference blocks exist exist before it. From a local state $s_j^i$, one can reconstruct a unique DAG. That is, the mapping from a local state  $s_j^i$ into a DAG is \emph{injective} or one-to-one. 
Thus, for Lachesis, we can simply denote the $j$-th local state of a process $i$ by the OPERA chain $g_j^i$ (often we simply use $G_i$ to denote the current local state of a process $i$).

\begin{defn}[action]
	An action is a function from one local state to another local state.
\end{defn}
Generally speaking, an action can be either: a $send(m)$ action where $m$ is a message, a $receive(m)$ action, and an internal action. A message $m$ is a triple $\langle i,j,B \rangle$ where $i \in P$ is the sender of the message, $j \in P$ is the message recipient, and $B$ is the body of the message. Let $M$ denote the set of messages. 
In Lachesis protocol, $B$ consists of the content of an event block $v$. 
Semantics-wise, in Lachesis, there are  two actions that can change a process's local state: creating a new event and receiving an event from another process.

\begin{defn}[event] An event is a tuple $\langle  s,\alpha,s' \rangle$ consisting of a state, an action, and a state.
\end{defn}

Sometimes, the event can be represented by the end state $s'$.
The $j$-th event in history $h_i$ of process $i$ is $\langle  s_{j-1}^i,\alpha,s_j^i \rangle$, denoted by $v_j^i$.

\begin{defn}[local history] A local history $h_i$ of process $i$ is a (possibly infinite) sequence of alternating local states  --- beginning with a distinguished initial state. A set $H_i$ of possible local histories for each process $i$ in $P$.
\end{defn}

The state of a process can be obtained from its initial state and the sequence of actions or events that have occurred up to the current state. 
In Lachesis protocol, we use append-only sematics. The local history may be equivalently described as either of the following:
$$h_i = s_0^i,\alpha_1^i,\alpha_2^i, \alpha_3^i \dots $$
$$h_i = s_0^i, v_1^i,v_2^i, v_3^i \dots $$
$$h_i = s_0^i, s_1^i, s_2^i, s_3^i, \dots$$

In Lachesis, a local history is equivalently expressed as:
$$h_i = g_0^i, g_1^i, g_2^i, g_3^i, \dots$$
where $g_j^i$ is the $j$-th local OPERA chain (local state) of the process $i$.

\begin{defn}[run] Each asynchronous run is a vector of local histories. Denoted by
	$\sigma = \langle h_1,h_2,h_3,...h_N \rangle$.
\end{defn}

Let $\Sigma$ denote the set of asynchronous runs.

We can now use Lamport’s theory to talk about global states of an asynchronous system.
A global state of run $\sigma$ is an $n$-vector of prefixes of local histories of $\sigma$, one prefix per process.
The happens-before relation can be used to define a consistent global state, often termed a consistent cut, as follows.

\begin{defn}[Consistent cut] A consistent cut of a run $\sigma$ is any global state such that if $v_x^i \rightarrow v_y^j$ and $v_y^j$ is in the global state, then $v_x^i$ is also in the global state. Denoted by $\vec{c}(\sigma)$.
\end{defn}

By Theorem~\ref{thm:conchains}, all nodes have consistent local OPERA chains. The concept of consistent cut formalizes such a global state of a run. A consistent cut consists of all consistent OPERA chains. A received event block exists in the global state implies the existence of the original event block.
Note that a consistent cut is simply a vector of local states; we will use the notation $\vec{c}(\sigma)[i]$ to indicate the local state of $i$ in cut $\vec{c}$ of run $\sigma$.



The formal semantics of an asynchronous system is given via  the satisfaction relation $\vdash$. Intuitively $\vec{c}(\sigma) \vdash \phi$, ``$\vec{c}(\sigma)$ satisfies $\phi$,'' if fact $\phi$ is true in cut $\vec{c}$ of run $\sigma$. 
We assume that we are given a function $\pi$ that assigns a truth value to each primitive proposition $p$. The truth of a primitive proposition $p$ in $\vec{c}(\sigma)$ is determined by $\pi$ and $\vec{c}$. This defines $\vec{c}(\sigma) \vdash p$.

\begin{defn}[equivalent cuts]
	Two cuts $\vec{c}(\sigma)$ and $\vec{c'}(\sigma')$ are equivalent  with respect to $i$ if: $$\vec{c}(\sigma) \sim_i \vec{c'}(\sigma') \Leftrightarrow \vec{c}(\sigma)[i] = \vec{c'}(\sigma')[i]$$
\end{defn}

We introduce two families of modal operators, denoted by $K_i$ and $P_i$, respectively. Each family indexed by process identifiers. 
Given a fact $\phi$, the modal operators are defined as follows:
\begin{defn}[$i$ knows $\phi$]
	$K_i(\phi)$ represents the statement ``$\phi$ is true in all possible consistent global states that include $i$’s local state''. 
		$$\vec{c}(\sigma) \vdash K_i(\phi) \Leftrightarrow \forall \vec{c'}(\sigma')   (\vec{c'}(\sigma') \sim_i \vec{c}(\sigma) \ \Rightarrow\ \vec{c'}(\sigma') \vdash \phi) $$
\end{defn}

\begin{defn}[$i$ partially knows $\phi$]
	$P_i(\phi)$ represents the statement ``there is some consistent global state in this run that includes $i$’s local state, in which $\phi$ is true.''
		$$\vec{c}(\sigma) \vdash P_i(\phi) \Leftrightarrow \exists \vec{c'}(\sigma) ( \vec{c'}(\sigma) \sim_i \vec{c}(\sigma) \ \wedge\ \vec{c'}(\sigma) \vdash \phi )$$ 
\end{defn}

The next modal operator is written $M^C$ and stands for ``majority concurrently knows.'' This is adapted from the ``everyone concurrently knows'' in CCK paper~\cite{cck92}.
The definition of $M^C(\phi)$ is as follows.

\begin{defn}[majority concurrently knows]
	$$M^C(\phi) =_{def} \bigwedge_{i \in S} K_i P_i(\phi), $$ where $S \subseteq P$ and $|S| > 2n/3$.	
\end{defn}

In the presence of one-third of faulty nodes, the original operator ``everyone concurrently knows'' is sometimes not feasible.
Our modal operator $M^C(\phi)$ fits precisely the semantics for BFT systems, in which unreliable processes may exist.

The last modal operator is concurrent common knowledge (CCK), denoted by $C^C$.
\begin{defn}[concurrent common knowledge]
	$C^C(\phi)$ is defined as a fixed point of $M^C(\phi \wedge X)$
\end{defn}
CCK defines a state of process knowledge that implies that all processes are in that same state of knowledge, with respect to $\phi$, along some cut of the run. In other words, we want a state of knowledge $X$ satisfying: $X = M^C(\phi \wedge X)$.	
$C^C$ will be defined semantically as the weakest such fixed point, namely as the greatest fixed-point of $M^C(\phi \wedge X)$.
It therefore satisfies:
$$C^C(\phi) \Leftrightarrow  M^C(\phi \wedge C^C(\phi))$$

Thus, $P_i(\phi)$ states that there is some cut in the same asynchronous run $\sigma$ including $i$’s local state, such that $\phi$ is true in that cut.

Note that $\phi$ implies $P_i(\phi)$. But it is not the case, in general, that $P_i(\phi)$ implies $\phi$ or even that $M^C(\phi)$ implies $\phi$. The truth of $M^C(\phi)$ is determined with respect to some cut $\vec{c}(\sigma)$. A process cannot distinguish which cut, of the perhaps many cuts that are in the run and consistent with its local state, satisfies $\phi$; it can only know the existence of such a cut. 

\begin{defn}[global fact]
	Fact $\phi$ is valid in system $\Sigma$, denoted by $\Sigma \vdash \phi$, if $\phi$ is true in all cuts of all runs of $\Sigma$.	
	$$\Sigma \vdash \phi 
	\Leftrightarrow (\forall \sigma \in \Sigma)(\forall\vec{c}) (\vec{c}(a) \vdash \phi)$$
\end{defn} 

\begin{defn}
	Fact $\phi$ is valid, denoted $\vdash \phi$, if $\phi$ is valid in all systems, i.e. 
	$(\forall \Sigma) (\Sigma \vdash \phi)$.
\end{defn}

\begin{defn}[local fact]
	 A fact $\phi$ is local to process $i$ in system $\Sigma$ if
	 $\Sigma \vdash (\phi \Rightarrow K_i \phi)$
\end{defn}

\begin{thm} If $\phi$ is local to process $i$ in system $\Sigma$, then $\Sigma \vdash (P_i(\phi) \Rightarrow \phi)$.	
\end{thm}

\begin{lem}
	If fact $\phi$ is local to 2/3 of the processes in a system $\Sigma$, then $\Sigma \vdash (M^C(\phi) \Rightarrow \phi)$ and furthermore $\Sigma \vdash (C^C(\phi) \Rightarrow \phi)$.
	\end{lem}

\begin{defn}
	A fact $\phi$ is attained in run $\sigma$ if $\exists \vec{c}(\sigma) (\vec{c}(\sigma) \vdash \phi)$.
\end{defn}

Often, we refer to ``knowing'' a fact $\phi$ in a state rather than in a consistent cut, since knowledge is dependent only on the local state of a process.
Formally, $i$ knows $\phi$ in state $s$ is shorthand for
$$\forall \vec{c}(\sigma) (\vec{c}(\sigma)[i] = s \Rightarrow \vec{c}(\sigma) \vdash \phi)$$

For example, if a process in Lachesis protocol knows a fork exists (i.e., $\phi$ is the exsistenc of fork) in its local state $s$ (i.e., $g_j^i$), then a consistent cut contains the state $s$ will know the existence of that fork.

\begin{defn}[$i$ learns $\phi$]
 Process $i$ learns $\phi$ in state $s_j^i$ of run $\sigma$ if $i$ knows $\phi$ in $s_j^i$ and, for all previous states $s_k^i$ in run $\sigma$, $k < j$, $i$ does not know $\phi$.
 \end{defn}

The following theorem says that if $C_C(\phi$ is attained in a run then all processes $i$ learn $P_i C^C(\phi)$ along a single consistent cut.

\begin{thm}[attainment]
	 If $C^C(\phi)$ is attained in a run $\sigma$, then the set of states in which all processes learn $P_i C^C(\phi)$ forms a consistent cut in $\sigma$.	
\end{thm}

We have presented a formal semantics of Lachesis protocol based on the concepts and notations of concurrent common knowledge~\cite{cck92}.
For a proof of the above theorems and lemmas in this Section, we can use similar proofs as described in the original CCK paper.

With the formal semantics of Lachesis, the theorems and lemmas described in Section~\ref{se:proof} can be expressed in term of CCK. For example, one can study a fact $\phi$ (or some primitive proposition $p$) in the following forms: `is there any existence of fork?'. One can make Lachesis-specific questions like 'is event block $v$ a root?', 'is $v$ a clotho?', or 'is $v$ a atropos?'. This is a remarkable result, since we are the first that define such a formal semantics for DAG-based protocol.




\clearpage
\section{Reference}\label{se:ref}

\renewcommand\refname{\vskip -1cm}
\bibliographystyle{unsrt}
\bibliography{LCA}

\end{document}